\documentclass[11pt,a4paper]{article}

\usepackage[utf8]{inputenc}
\usepackage[T1]{fontenc}
\usepackage{lmodern}
\usepackage[margin=1in]{geometry}
\usepackage{amsmath,amssymb,amsthm,mathtools}
\usepackage{microtype}
\usepackage{needspace}
\usepackage{booktabs}
\usepackage{enumitem}
\usepackage{float}
\newif\ifreviewalg
\IfFileExists{algorithm.sty}{\reviewalgtrue}{\reviewalgfalse}
\ifreviewalg
  \usepackage{algorithm}
  \usepackage[noend]{algpseudocode}
\else
  \floatstyle{ruled}\newfloat{algorithm}{htbp}{loa}\floatname{algorithm}{Algorithm}
  \newcounter{reviewalgdepth}
  \newenvironment{algorithmic}[1][1]{%
    \setcounter{reviewalgdepth}{0}%
    \begin{enumerate}[label=\arabic*:,leftmargin=2.2em,labelsep=.6em,itemsep=2pt,parsep=0pt,topsep=3pt]\small
  }{\end{enumerate}}
  \newcommand{\State}{\item\hspace*{\value{reviewalgdepth}em}}

  \newcommand{\For}[1]{\State\textbf{for} #1 \textbf{do}\stepcounter{reviewalgdepth}}
  \newcommand{\EndFor}{\addtocounter{reviewalgdepth}{-1}}
  
  \newcommand{\Return}{\textbf{return}\ }
\fi

\usepackage{tikz}
\usetikzlibrary{arrows.meta,positioning}
\usepackage{xcolor}
\usepackage[colorlinks=true,linkcolor=blue!50!black,citecolor=green!40!black,urlcolor=blue!50!black]{hyperref}
\usepackage[capitalise,nameinlink]{cleveref}
\hypersetup{hypertexnames=false}
\usepackage{thmtools}

\declaretheorem[name=Theorem,numberwithin=section]{theorem}
\declaretheorem[name=Lemma,sibling=theorem]{lemma}
\declaretheorem[name=Corollary,sibling=theorem]{corollary}

\declaretheorem[name=Fact,sibling=theorem]{fact}

\declaretheorem[name=Remark,sibling=theorem,style=remark]{remark}
\crefname{algorithm}{Algorithm}{Algorithms}

\usepackage{mathrsfs}
\newcommand{\eps}{\varepsilon}
\newcommand{\gf}{\gamma_{f}}
\newcommand{\N}{\mathbb{N}}
\newcommand{\R}{\mathbb{R}}
\newcommand{\E}{\mathbb{E}}
\newcommand{\Pp}{\mathcal{P}}
\newcommand{\Cc}{\mathscr{C}}
\newcommand{\Oof}{\mathcal{O}}
\newcommand{\id}{\mathrm{id}}
\newcommand{\lead}{\mathrm{lead}}
\newcommand{\cut}{\mathrm{cut}}
\newcommand{\LP}{\mathrm{LP}}
\newcommand{\OPT}{\mathrm{OPT}}
\newcommand{\HEAD}{\textsc{head}}
\newcommand{\TAIL}{\textsc{tail}}
\DeclareMathOperator{\dist}{dist}
\definecolor{saeedcolor}{rgb}{0.80,0.00,0.00}
\definecolor{sebicolor}{rgb}{0.00,0.50,0.15}

\title{Almost Optimal Constant-Round Approximation of \\Dominating Set in Graph Classes with Excluded Minors}

\author{
  Saeed Amiri\\
  Independent researcher
  \and
  Sebastian Siebertz\\
  University of Bremen
}

\date{\today}

\begin{document}
\maketitle

\begin{abstract}
For every fixed proper minor-closed class $\Cc$ and every $\eps>0$, we give a deterministic LOCAL algorithm that returns a dominating set of size at most $(2a(\Cc)+1+\eps)\gf(G)$ on every $G\in\Cc$.
Here $a(\Cc)$ is the supremum edge-to-vertex ratio in~$\Cc$, and $\gf(G)$ is the fractional domination number.
The class also admits a deterministic $(1+\eps)$-approximation for fractional dominating set and a randomized algorithm that always returns a dominating set and has expected size at most $(1+\eps)\gamma(G)$.
In each case, the number of rounds depends only on $\Cc$ and~$\eps$.
None of these algorithms requires the number of vertices or the maximum degree as part of the input.

For planar graphs, this gives the deterministic guarantee $(7+\eps)\gf(G)$.
Together with the lower bound of Hilke, Lenzen and Suomela, it determines the infimum of the deterministic constant-round approximation ratios for planar minimum dominating set as~$7$, settling a question that had remained open since their work.
The corresponding infima, measured against the integral optimum, are $7$ for graphs of Euler genus at most any fixed $g\ge0$, $2t-3$ for $K_t$-minor-free graphs with $3\le t\le9$, and $2r+1$ for graphs of treewidth or pathwidth at most any fixed $r\ge1$.
We also prove that, for every integer~\mbox{$r\ge1$}, no deterministic constant-round LOCAL algorithm achieves an approximation ratio below~\mbox{$2r+1$} on the $r$-th powers of paths, even when every vertex knows the number of vertices.
This gives a new proof that the limiting constants are optimal for planar graphs, graphs of bounded treewidth or pathwidth, and $K_t$-minor-free graphs with $3\le t\le9$. 
For triangle-free planar graphs, the corresponding infimum is $5$.
\end{abstract}

\paragraph{Disclosure of AI use.}
We have long been working on combining cluster-and-solve
algorithms~\cite{CH06,CHW08}, randomized star contraction~\cite{LR15},
Kuhn's derandomization principle for distributed linear
programs~\cite{Kuh05}, and local threshold rounding~\cite{BU17,Dvo19,HKOSV25}.
The recent result of Wawrzyniak~\cite{Waw26} motivated us to complete this work.
We then used AI agents to fill in technical details, polish the language,
and check the correctness of the proofs.

\section{Introduction}\label{sec:intro}

A \emph{dominating set} of a graph $G$ is a set $D\subseteq V(G)$ such that every vertex belongs to $D$ or has a neighbor in~$D$.
The minimum size of such a set is the \emph{domination number} $\gamma(G)$.
Finding a minimum dominating set is NP-hard even on planar graphs~\cite[Problem GT2]{GareyJohnson1979}.
In a distributed network, a dominating set can serve as a set of local coordinators, with every vertex either acting as a coordinator or being adjacent to one.
We study the problem in the LOCAL model~\cite{Linial92,Peleg00}, where the input graph is the communication network, vertices have distinct identifiers, and communication proceeds in synchronous rounds, with no bound on message size or limits on local computation.
In a constant number of rounds, a vertex can base its decision only on information within a bounded distance.
Algorithms with this restriction are usually called \emph{local algorithms}~\cite{NS95,Suomela13}.

On general graphs, a constant number of rounds does not suffice for a constant-factor approximation~\cite{KMW04,KMW16}.
Planar graphs admit such algorithms, and determining the best approximation ratio has taken a succession of improvements, as shown in \cref{tab:history}.
Lenzen, Oswald and Wattenhofer~\cite{LOW08} introduced the first constant-factor algorithm.
The full treatment by Lenzen, Pignolet and Wattenhofer~\cite{LPW13} gives ratio $130$, and Wawrzyniak~\cite{Waw14} improved the analysis to obtain $52$.
Heydt, Siebertz and Vigny obtained ratio $20$ in their 2021 preprint~\cite{HSV21} and improved it to $11+\eps$ in the published conference version~\cite{HSV22}.
Their subsequent joint work with Kublenz and Ossona de Mendez~\cite{HKOSV25} generalized the approach to sparse graph classes.
Most recently, Wawrzyniak~\cite{Waw26} obtained ratio $8+\eps$.
The constant $7$ was already within reach of earlier techniques, but only with
a round bound depending on the maximum degree; see \cref{sec:related}.
The $(11+\eps)$- and $(8+\eps)$-algorithms use a three-phase approach whose final phase rounds a fractional solution on a residual instance of bounded degree.
The lower bounds have a different origin: a Ramsey argument of Czygrinow, Ha{\'n}{\'c}kowiak and Wawrzyniak~\cite{CHW08} excludes every ratio below $5$, and a triangular-grid construction of Hilke, Lenzen and Suomela~\cite{HLS14} excludes every ratio below $7$.
Thus the remaining gap was between the lower bound $7$ and algorithms with ratio $8+\eps$.

We close this gap at the level of the infimum as a consequence of a general theorem for proper minor-closed classes.
These are the classes closed under taking minors that exclude at least one graph, and they include planar graphs, graphs of bounded Euler genus, and graphs of bounded treewidth or pathwidth.
For each such class, our approximation guarantee is determined by its edge density.
The resulting bounds have optimal limiting constants not only for planar graphs, but also for every fixed Euler-genus bound, for $K_t$-minor-free graphs with $3\le t\le9$, and for every fixed treewidth or pathwidth bound.

\begin{table}[t]
\centering\small
\begin{tabular}{@{}lp{0.56\textwidth}p{0.24\textwidth}@{}}
\toprule
ratio or barrier & reference & remark\\
\midrule
$130$ & Lenzen, Pignolet, Wattenhofer~\cite{LPW13} & full treatment of~\cite{LOW08}, first constant ratio\\
$52$ & Wawrzyniak~\cite{Waw14} & new analysis of~\cite{LPW13}\\
$11+\eps$ & Heydt, Siebertz, Vigny~\cite{HSV22} & LP in last phase\\
$8+\eps$ & Wawrzyniak~\cite{Waw26} & framework of~\cite{HKOSV25}\\
$7+\eps$ & this paper & against the fractional optimum\\
\midrule
$5$ & Czygrinow, Ha{\'n}{\'c}kowiak, Wawrzyniak~\cite{CHW08} & every smaller ratio excluded\\
$7$ & Hilke, Lenzen, Suomela~\cite{HLS14} & every smaller ratio excluded\\
\bottomrule
\end{tabular}
\caption{Deterministic constant-round approximation of minimum dominating set on planar graphs in the LOCAL model.}
\label{tab:history}
\end{table}

\subsection{Our results}

Let $\gf(G)$ denote the minimum weight of a fractional dominating set of~$G$.
Equivalently, $\gf(G)$ is the optimum value of the dominating-set linear program, whose formulation is given in \cref{sec:prelim}.
Since $\gf(G)\le\gamma(G)$, an approximation guarantee against $\gf(G)$ is at least as strong as the corresponding guarantee against~$\gamma(G)$.
For a proper minor-closed class $\Cc$, let~$a(\Cc)$ be the supremum of $|E(G)|/|V(G)|$ over its nonempty members.
If $\Cc$ excludes $K_t$ as a minor, then
$a(\Cc)=\Oof(t\sqrt{\log t})$ by the extremal bounds of
Kostochka and Thomason~\cite{Kos84,Tho01}.
More precisely, the supremum of $|E(G)|/|V(G)|$ over all
nonempty $K_t$-minor-free graphs is
$(\alpha+o(1))t\sqrt{\ln t}$ as $t\to\infty$,
where $\alpha=0.319086\ldots$~\cite{Tho01}.
We prove the following main theorem. 
\pagebreak

\begin{theorem}\label{thm:main}
Let $\Cc$ be a fixed proper minor-closed class and let $\eps>0$.
The following algorithms exist in the LOCAL model, each with a number of rounds depending only on $\Cc$ and~$\eps$.
\vspace{-1mm}
\begin{enumerate}[label=(\alph*)]
\item A randomized algorithm that always returns a dominating set $D$ and satisfies $\E|D|\le(1+\eps)\gamma(G)$ on every $G\in\Cc$.
\item A deterministic algorithm that returns a fractional dominating set of weight at most $(1+\eps)\gf(G)$ on every $G\in\Cc$.
\item A deterministic algorithm that returns a dominating set of size at most $(2a(\Cc)+1+\eps)\gf(G)$ on every $G\in\Cc$.
\end{enumerate}
\vspace{-1mm}
None of these algorithms requires the number of vertices or the maximum degree with the input.
\end{theorem}

The planar case gives the first main consequence.
Planar graphs have $a(\Cc)=3$, since every planar graph has at most $3|V(G)|$ edges and the edge-to-vertex ratios of planar triangulations tend to~$3$.
Together with the lower bound of Hilke, Lenzen and Suomela~\cite{HLS14}, \cref{thm:main}(c) therefore gives the following corollary.

\Needspace{10\baselineskip}
\begin{corollary}\label{cor:planar-intro}\label{cor:optimal}
For every $\eps>0$, a deterministic $O_\eps(1)$-round LOCAL algorithm returns a dominating set of size at most $(7+\eps)\gf(G)$ on every planar graph~$G$.
The algorithm requires neither the number of vertices, nor the maximum degree, nor a planar embedding with the input.
Hence, the infimum of the approximation ratios achievable by deterministic constant-round LOCAL algorithms for planar minimum dominating set is~$7$.
\end{corollary}

For $\eps\in(0,1]$, our planar algorithm runs in $T_0+\Oof((c_0/\eps)^\kappa)$ rounds, where $c_0=52$, $T_0$ is the round complexity of the initial constant-factor algorithm, and $\kappa=\ln3/\ln(24/23)<25.9$.
The number of rounds is independent of the order and maximum degree of the input graph, but grows as $\eps$ tends to zero.
Whether a deterministic constant-round algorithm can attain ratio exactly $7$ remains open.

The dependence on $a(\Cc)$ can be improved when the density bound has a fixed additive excess.
Define the \emph{asymptotic density} $d^*(\Cc)$ to be the infimum of the numbers $d\ge0$ for which there is a constant $b\ge0$ such that $|E(G)|\le d|V(G)|+b$ for every $G\in\Cc$.
We always have $d^*(\Cc)\le a(\Cc)$.
For example, the class of graphs of Euler genus at most a fixed $g$ has asymptotic density~$3$, even though its supremum density can be larger.

\begin{theorem}\label{thm:asym-intro}
For every fixed proper minor-closed class $\Cc$ and every $\eps>0$, a deterministic $O_{\Cc,\eps}(1)$-round LOCAL algorithm returns a dominating set of size at most $(2d^*(\Cc)+1+\eps)\gamma(G)$ on every $G\in\Cc$.
\end{theorem}

Note that in \cref{thm:asym-intro} the approximation is with respect to the integral optimum $\gamma(G)$.
In contrast, \cref{thm:main}(c) gives the stronger guarantee against the fractional optimum, with~$a(\Cc)$ in place of $d^*(\Cc)$.
This yields the following consequences.

\begin{corollary}\label{cor:sharp-intro}
For every $\eps>0$, deterministic constant-round LOCAL algorithms have the following guarantees.
\begin{enumerate}[label=(\alph*)]
\item On graphs of Euler genus at most a fixed $g\ge0$, the computed dominating set has size at most $(7+\eps)\gamma(G)$.
For $g\le2$, the stronger bound $(7+\eps)\gf(G)$ holds.
\item On $K_t$-minor-free graphs, for each fixed $3\le t\le9$, the computed dominating set has size at most $(2t-3+\eps)\gf(G)$.
\item On graphs of treewidth or pathwidth at most a fixed $r\ge1$, the computed dominating set has size at most $(2r+1+\eps)\gf(G)$.
\end{enumerate}
Hence, the infima of the deterministic constant-round approximation ratios against the integral optimum are $7$, $2t-3$, and $2r+1$, respectively.
\end{corollary}

Thus the planar limiting constant remains optimal for every fixed Euler-genus bound, improving the previous genus-independent ratio $25+\eps$~\cite{Waw26,BDGPSW26}.
We also prove matching lower bounds for the excluded-clique and width bounds.
For every $r\ge1$, we prove that no deterministic constant-round algorithm achieves a ratio below $2r+1$ on the $r$-th powers of paths.
These graphs have pathwidth $r$ once their order is at least $r+1$, and exclude $K_{r+2}$ as a minor.
For $r=3$, these graphs are planar and have pathwidth three, so the same argument recovers the planar lower bound on this restricted family. The lower bounds hold even for algorithms in which every vertex knows the number of vertices.

For triangle-free planar graphs, \cref{thm:main}(b) and a sharper rounding bound give output size at most $(5+\eps)\gf(G)$, improving the previous ratio $6$~\cite{Waw25}.
Together with the square-grid lower bound~\cite[Section~4]{BCGW21}, this determines the infimum of the deterministic constant-round approximation ratios as $5$.
More generally, a planar graph of girth at least $g\ge3$ admits the guarantee $(2g/(g-2)+1+\eps)\gf(G)$.
We also obtain bounds for $K_{2,t}$-minor-free graphs and $K_{3,3}$-minor-free graphs, and asymptotic bounds for exclusion of an arbitrary clique.
Our results are summarized in \cref{tab:beyond}.

\begin{table}[t]
\centering\small
\begin{tabular}{@{}p{0.28\textwidth}p{0.32\textwidth}p{0.16\textwidth}p{0.14\textwidth}@{}}
\toprule
class & earlier guarantees & this paper & lower bound\\
\midrule
planar graphs & $8+\eps$~\cite{Waw26} & $7+\eps$ & $7-\eps$~\cite{HLS14}\\
triangle-free planar graphs & $6$~\cite{Waw25} & $5+\eps$ & $5-\eps$~\cite{BCGW21}\\
Euler genus $\le g$ & $25+\eps$ (genus-independent)~\cite{Waw26,BDGPSW26} & $7+\eps$ & $7-\eps$~\cite{HLS14}\\
$K_t$-minor-free, $4\le t\le9$ & $t^{\Oof(t^2\sqrt{\log t})}$~\cite{KSV21} & $2t-3+\eps$ & $2t-3-\eps$\\
$K_t$-minor-free, $t\ge10$ & $t^{\Oof(t^2\sqrt{\log t})}$~\cite{KSV21} & $2a_t+1+\eps$ & $2t-3-\eps$\\
treewidth or pathwidth $\le r$, $r\ge1$ & via $K_{r+2}$-minor-free graphs & $2r+1+\eps$ & $2r+1-\eps$\\
$K_{2,3}$-minor-free & $5$~\cite[Thm.~4.4]{BGPW25b} & $5+\eps$ & $5-\eps$~\cite{BCGW21}\\
$K_{2,t}$-minor-free, $t\ge4$ & $\min\{2t-1,50\}$~\cite{BGPW25b} & $t+2+\eps$ & $5-\eps$~\cite{BCGW21}\\
$K_{3,3}$-minor-free & $14+\eps$~\cite{BGPW25b,HKOSV25} & $7+\eps$ & $7-\eps$~\cite{HLS14}\\
proper minor-closed $\Cc$ & $O_\Cc(1)$~\cite{CHWW18,HKOSV25} & $2d^*(\Cc)+1+\eps$ & \\
\bottomrule
\end{tabular}
\caption{Deterministic constant-round approximation guarantees for minimum dominating set.
All ratios in the table are measured against $\gamma$.
Lower bounds without citations follow from \cref{thm:path-lb}.
For $K_{2,t}$-minor-free graphs, combining the available algorithms gives $\min\{t+2+\eps$, $2t-1,50\}$; for sufficiently small positive $\eps$, the new bound improves the displayed earlier guarantees exactly when $4\le t\le47$.
The $K_{3,3}$ bound $14+\eps$ is the case $t=3$, after rescaling $\eps$, of the guarantee $(2+\eps)(t+4)$ reported in~\cite{BGPW25b} from~\cite{HKOSV25}.
Here $a_t$ is the supremum edge density of $K_t$-minor-free graphs; the extremal bounds and the path powers give $a_t=t-2$ for $3\le t\le9$, while $a_t=\Theta(t\sqrt{\log t})$ in general~\cite{Kos84,Tho01}.}
\label{tab:beyond}
\end{table}

\subsection{Proof overview}

Our starting point is the approximation scheme of Czygrinow, Ha{\'n}{\'c}kowiak and Wawrzyniak~\cite[Theorem~3.4]{CHW08} and its extension to excluded-minor classes~\cite{CHWW18}.
Fix a proper minor-closed class~$\Cc$ with $a=a(\Cc)>0$.
We first compute a dominating set $D_0$ of size at most $c\gamma(G)$, where $c$ depends only on~$\Cc$, and partition the vertices into connected groups, one around each vertex of~$D_0$.
Contracting these groups gives a simple graph $H\in\Cc$ with at most $a|D_0|$ edges.
The algorithm then partitions $H$ into clusters of bounded diameter with few edges between clusters, and solves the domination problem inside the corresponding clusters of the original graph.
The extra cost is at most twice the number of edges between clusters in~$H$.
In the deterministic scheme there is a dependence on $n$ that comes from the clustering step: star selection uses a proper coloring of a pseudo-forest with a constant number of colors, obtained in $\Oof(\log^*n)$ rounds by the Cole-Vishkin algorithm~\cite{CV86}.
Known lower bounds prevent a deterministic approximation scheme from having constant round complexity~\cite{CHW08,LW08}.
\pagebreak

For the clustering step, we instead use a variant of the randomized
heavy-star contraction described by Levi and Ron~\cite{LR15},
building on the clustering method of Czygrinow, Ha{\'n}{\'c}kowiak
and Wawrzyniak~\cite{CHW08}.
Each current part tosses one fair coin at each level, and the
coin outcomes determine which of the selected edges are contracted.
Every quotient remains in $\Cc$, so the same density bound
applies at every level.
For our variant, conditioned on the preceding levels, the
remaining edge weight decreases in expectation by a factor of
at most $q=1-1/(8a)$.
Applied to the auxiliary graph obtained from $D_0$, this gives
an expected boundary cost of at most $2acq^L\gamma(G)$ after
$L$ levels.
For $0<\eps\le1$, choosing $L=O_{\Cc}(1+\log(1/\eps))$
makes this cost at most $\eps\gamma(G)$.
The cluster diameters depend only on $L$, which gives the
randomized constant-round algorithm in \cref{thm:main}(a).
For planar graphs, $a=3$ and $q=23/24$.

We remove the randomness using Kuhn's derandomization
principle for distributed linear programs~\cite[Theorem~2.32]{Kuh05}.
For the integral variant, this means setting $x_v$ to the
probability that the randomized algorithm selects~$v$.
For each fixed input, this probability depends on a neighborhood of bounded radius and finitely many fair coins, so $v$ can compute it exactly by enumerating their outcomes.
Every outcome is a dominating set, and taking expectations in each domination constraint shows that the vector $x$ is a fractional dominating set.
Moreover, $|x|$ equals the expected size of the randomized output.
Threshold rounding~\cite{BU17,Dvo19,HKOSV25} then returns a dominating set of size at most $(2a+1)|x|$ in one additional round, when the threshold can be compared exactly.
For planar graphs, the factor is~$7$.
More generally, a slightly larger rational density parameter allows exact comparisons while changing the final guarantee by an arbitrarily small amount.

Averaging the integral solutions gives a weight bound in terms of $\gamma(G)$.
To obtain the stronger guarantee against $\gf(G)$, we solve
the fractional domination problem inside each cluster and
apply the same averaging principle to these rational solutions.
The rounding lemma also implies $\gamma(G)\le(2a+1)\gf(G)$, so the expected boundary cost is at most $2a(2a+1)cq^L\gf(G)$.
Choosing~$L$ accordingly gives a fractional solution of weight at most $(1+\eps)\gf(G)$, proving \cref{thm:main}(b).
Applying threshold rounding with a suitably smaller error proves \cref{thm:main}(c).
The global optimum $\gf(G)$ serves only as a benchmark in this argument.
The algorithm solves LPs inside bounded-diameter clusters, and never needs to compute the global optimum.

For \cref{thm:asym-intro}, we refine the rounding argument to allow a bounded additive excess in the density bound.
This gives an additive constant in the output size.
Following the additive-error reduction of~\cite[full version, Proposition~D.2]{BDGPSW26}, we solve components of sufficiently small diameter exactly and absorb the additive cost on every remaining component using $\gamma(K)\ge(\operatorname{diam}(K)+1)/3$ for every nonempty connected graph~$K$.
This step gives the guarantee against the integral optimum and yields ratio $7+\eps$ for every fixed Euler genus.

\subsection{Related work}\label{sec:related}

The cluster-and-solve approach of Czygrinow and Ha{\'n}{\'c}kowiak~\cite{CH06} and Czygrinow, Ha{\'n}{\'c}kowiak and Wawrzyniak~\cite{CHW08} builds on earlier uses of separators and bounded-treewidth decompositions for planar optimization~\cite{LT80,Baker94}.
Czygrinow et al.~\cite{CHW08} already used randomization to obtain a constant-round approximation scheme for maximum independent set that succeeded with high probability, and discussed extensions to other problems.
Randomized star contraction with one fair coin per part also
appeared in the partition-oracle construction of Levi and
Ron~\cite{LR15} for excluded-minor graphs of bounded degree.
Our clustering uses a variant of this rule together with the
heaviest-edge pseudo-forest of~\cite[Fact~2.2]{CHW08}.
We give the analysis for dominating set directly on minor-closed
classes in \cref{sec:cluster}, keeping track of the expected
boundary cost and the local randomness needed for the subsequent
averaging step.

Approximating a covering LP and rounding it locally is also an established approach.
Kuhn and Wattenhofer~\cite{KW05} gave the first constant-round distributed algorithm for minimum dominating set on general graphs with a non-trivial expected approximation ratio, using randomized rounding of a locally computed LP solution and obtaining a ratio that depends on the maximum degree.
The covering-LP algorithm of Kuhn, Moscibroda and Wattenhofer~\cite{KMW06,KMW16}, specialized to the unweighted dominating-set LP, gives a deterministic $(1+\eps)$-approximation in $\Oof(\log^2(\max\{2,\Delta\})/\eps^4)$ rounds.
Together with the rounding bound of Bansal and Umboh~\cite{BU17}, improved by Dvo{\v{r}}{\'a}k~\cite{Dvo19}, this gave ratio $(2\alpha+1)(1+\eps)$ on graphs of arboricity at most~$\alpha$.
Dory, Ghaffari and Ilchi~\cite[Theorem~1]{DGI22} obtained this ratio deterministically in $\Oof((1+\log^+(\Delta/\alpha))/\eps)$ CONGEST rounds, even for weighted domination, assuming that $\Delta$ and $\alpha$ are known to all vertices. Here $\log^+z=\max\{0,\log z\}$.
For planar graphs, one can take $\alpha=3$ and rescale $\eps$ to obtain ratio $7+\eps$, but the round bound still depends on~$\Delta$.
Their lower bound on graphs of arboricity $2$ showed that some dependence on degree was unavoidable in that broader class: every constant-factor approximation required $\Omega(\log\Delta/\log\log\Delta)$ rounds on suitable such graph families.
Our use of sparse contractions permits a bound independent of degree on every fixed proper minor-closed class.

Other work addressed restricted message sizes.
Amiri~\cite{Ami21} gave an $\Oof(\alpha^2)$-approximation in $\Oof(\log n)$ CONGEST rounds.
Alipour and Farokhnejad~\cite{AF26} obtained ratio $\Oof(\alpha\log\Delta/\log\log\Delta)$ in $\Oof(\log\Delta/\log\log\Delta)$ rounds with one-bit messages, without prior knowledge of~$\alpha$ but with prior knowledge of~$\Delta$. As usual, bounded values of $\Delta$ receive a constant-size treatment.
In the planar algorithms of Heydt et al.~\cite{HKOSV25} and Wawrzyniak~\cite{Waw26}, fractional rounding was applied only after earlier phases had reduced the remaining instance to bounded degree.
The cost of those phases accounted for the difference between their guarantees and the factor $7$ from rounding.
\vspace{-2mm}

\paragraph{Subclasses of planar graphs.}
For outerplanar graphs, Bonamy, Cook, Groenland and Wesolek~\cite{BCGW21} gave a deterministic constant-round $5$-approximation and proved that no smaller ratio was possible.
For triangle-free planar graphs, Wawrzyniak~\cite{Waw25} obtained ratio $6$, improving the earlier bound $32$ of Alipour, Futuhi and Karimi~\cite{AFK20}.
Square-grid instances give the lower bound $5$ for triangle-free planar graphs~\cite[Section~4]{BCGW21}, as discussed in \cref{sec:subclasses}.
For planar graphs of girth at least $5$, Alipour and Jafari~\cite{AJ20} obtained ratio~$18$.
\vspace{-2mm}

\paragraph{Excluded minors and bounded expansion.}
Constant-factor constant-round algorithms are known for bounded-genus graphs~\cite{ASS16,ASS19,CHWW19}, graphs excluding a fixed topological minor~\cite{CHWW18}, and classes of bounded expansion~\cite{KSV21,HKOSV25}.
Amiri, Ossona de Mendez, Rabinovich and Siebertz~\cite{AORS18} also gave $\Oof(\log n)$-round constant-factor algorithms on bounded-expansion classes, including for distance-$r$ domination, and Amiri and
Wiederhake~\cite{AmiriW22} studied distance-$r$ covering problems on sparse
graphs of high girth.
The deterministic approximation scheme of~\cite{CHW08} was extended to bounded genus by Amiri and Schmid~\cite{AS16}, to sub-logarithmic expansion by Amiri, Schmid and Siebertz~\cite{ASS19}, and to excluded-minor classes by Czygrinow et al.~\cite{CHWW18}.
Bonamy, Gavoille, Picavet and Wesolek~\cite{BGPW25} obtained a constant approximation ratio independent of the Euler-genus bound.
Bonamy et al.~\cite{BDGPSW26} developed a general transfer theorem from planar graphs to bounded Euler genus and other classes with suitable local structure.
For $K_{2,t}$-minor-free graphs, Bonamy, Gavoille, Picavet and Wesolek~\cite{BGPW25b} gave a constant-round $50$-approximation and a three-round $(2t-1)$-approximation.
We compare these bounds with our results in \cref{sec:beyond}.

Without a restriction on the number of communication rounds, minimum dominating set admits a polynomial-time approximation scheme on planar graphs by Baker\textquotesingle s decomposition~\cite{Baker94}, and on $H$-minor-free graphs for every fixed $H$ by Grohe\textquotesingle s decomposition into graphs of almost bounded local treewidth~\cite{Grohe03}.

Feige, Hajiaghayi and Lee~\cite{FHL08} also gave a constant-factor approximation for treewidth on these graphs and used it to obtain polynomial-time approximation schemes for several optimization problems.
No such scheme exists in a constant number of rounds: by \cref{thm:path-lb}, the
achievable ratio is bounded away from~$1$ already on graphs of pathwidth~$r$.
\vspace{-2mm}

\paragraph{Derandomization.}
Network decompositions provide general tools for derandomization
in the LOCAL model~\cite{GHK18,RG20}.
For distributed linear programs, Kuhn~\cite[Theorem~2.32]{Kuh05}
showed that an always-feasible randomized algorithm could be
derandomized by replacing each output coordinate with its
expectation.
He also explicitly noted the application to fractional
dominating set.
We apply this principle here.
Feasibility is preserved by the linear domination constraints,
and the objective is preserved by linearity of expectation.
Under our finite-randomness assumptions, each expected
coordinate is exactly computable from a bounded local view.
This transformation adds only a constant number of rounds.

\subsection{Organization}
In \cref{sec:prelim,sec:rounding}, we introduce the model and recall the rounding bound.
The randomized clustering and domination algorithms are presented in \cref{sec:cluster,sec:rand}.
\Cref{sec:derand} applies Kuhn's derandomization principle and proves \cref{thm:main}.
\Cref{sec:excess} establishes the refinement using asymptotic density.
In \cref{sec:beyond}, we derive the planar, bounded-genus and excluded-minor consequences and prove the lower bounds from powers of paths.
We conclude with open problems in \cref{sec:open}.


\section{Preliminaries}\label{sec:prelim}

All graphs are finite, simple and undirected unless stated otherwise.
For $v\in V(G)$, let $N(v)$ be its neighborhood and $N[v]=N(v)\cup\{v\}$.
For $S\subseteq V(G)$, put $N[S]=\bigcup_{v\in S}N[v]$ and let~$G[S]$ denote the subgraph induced by~$S$.
The distance $\dist_G(u,v)$ is the length of a shortest $u$-$v$ path, or $\infty$ if no such path exists.
We write $B_G(v,r)=\{u\in V(G):\dist_G(u,v)\le r\}$.
For $x\in\R^{V(G)}$, set $x(S)=\sum_{v\in S}x_v$ and $|x|=x(V(G))$.
For nonnegative edge weights $w$, define $w(F)=\sum_{e\in F}w(e)$ and $w(G)=w(E(G))$.
\medskip

A \emph{fractional dominating set} of $G$ is a vector $x\in\R_{\ge0}^{V(G)}$ satisfying $x(N[u])\ge1$ for every $u\in V(G)$.
Its weight is $|x|$, and the \emph{fractional domination number} is
\[
\gf(G)=\min\bigl\{|x|:x\ge0,\ x(N[u])\ge1\text{ for every }u\in V(G)\bigr\}.
\]
The indicator vector of any dominating set is feasible for this linear program, so $\gf(G)\le\gamma(G)$.
\medskip

For $a>0$, a graph $G$ is \emph{$a$-sparse} if $|E(H)|\le a|V(H)|$ for every subgraph $H$ of~$G$.
A \emph{fractional orientation} assigns to each edge $e=uv$ numbers $o_e(u),o_e(v)\ge0$ with $o_e(u)+o_e(v)=1$.
Here~$o_e(u)$ is the fraction of the edge directed from $u$ to~$v$.
The fractional out-degree of $u$ is $\sum_{e\ni u}o_e(u)$.
An ordinary orientation is the special case in which all these numbers belong to $\{0,1\}$.

\begin{fact}\label{fact:sparse}
Every planar graph is $3$-sparse, and every triangle-free planar graph and every outerplanar graph is $2$-sparse.
More generally, for every integer $g\ge3$, a planar graph of girth at least $g$ is $\frac{g}{g-2}$-sparse.
\end{fact}

The following standard characterization follows from linear-programming duality and Charikar's formulation of the densest-subgraph problem~\cite{Cha00}.
It is also stated in~\cite[Theorem~1]{CHR25}.

\begin{lemma}\label{lem:orientation}
For every $a>0$, every $a$-sparse graph has a fractional orientation with out-degree at most~$a$ at each vertex.
\end{lemma}

When $a$ is an integer, flow integrality gives an ordinary orientation with maximum out-degree at most~$a$, in agreement with Hakimi's theorem~\cite{Hakimi65}.
We use these orientations only to prove size bounds, the algorithms do not compute them.
\medskip

A graph $H$ is a \emph{minor} of a graph $G$ if it can be
obtained from $G$ by deleting vertices, deleting edges, and
contracting edges. After each contraction, loops are deleted
and parallel edges are replaced by a single edge.
A \emph{graph class} is \emph{minor-closed} if it contains
every minor of each of its members, and \emph{proper} if it
does not contain every graph.
In particular, a minor-closed class is closed under taking subgraphs.
\medskip

We use two density parameters:
\begin{align*}
a(\Cc)&=\sup\Bigl(\{0\}\cup\bigl\{|E(G)|/|V(G)|:G\in\Cc,\ V(G)\ne\emptyset\bigr\}\Bigr),\\
d^*(\Cc)&=\inf\bigl\{d\ge0:\text{there is $b\ge0$ such that $|E(G)|\le d|V(G)|+b$ for every $G\in\Cc$}\bigr\}.
\end{align*}

The second parameter is called the \emph{asymptotic density} of~$\Cc$.
We always have $d^*(\Cc)\le a(\Cc)$.
If $\Cc$ is closed under subgraphs and $a(\Cc)>0$, every member is $a(\Cc)$-sparse.
If $\Cc$ is closed under disjoint unions, the two densities agree.
Indeed, if $|E(G)|>d|V(G)|$, then the excess over~$d|V|$ in the disjoint union of $m$ copies of $G$ tends to infinity with~$m$.
The distinction matters for bounded genus: the class of graphs of Euler genus at most $g$ has asymptotic density~$3$, by \cref{lem:genus-density} and planar triangulations of unbounded order, whereas its supremum density $a$ is unbounded as $g$ increases.
The latter follows because every complete graph embeds in some surface.
The extremal bound for $K_t$-minor-free graphs gives $|E(G)|=\Oof(t\sqrt{\log t})|V(G)|$~\cite{Kos84,Tho01}, and hence $a(\Cc)<\infty$.
\medskip

Let $H$ be an edge-weighted graph, and let $\Pp$ be a partition of $V(H)$ into nonempty sets that induce connected subgraphs.
The \emph{weighted quotient} $H/\Pp$ has vertex set $\Pp$.
Distinct parts $P,Q$ are adjacent precisely when an edge of $H$ joins them, and the weight of $PQ$ is the sum of the weights of all such edges.
The underlying simple graph is a minor of~$H$, and
\[
w(H/\Pp)=\cut_w(\Pp)
\coloneqq\sum_{\substack{e\in E(H)\\e\text{ joins distinct parts of }\Pp}}w(e).
\]
Consequently, every minor-closed class containing $H$ also contains the underlying simple graph of $H/\Pp$.
\medskip

We now define the LOCAL model.
The input graph $G=(V,E)$ is the communication network.
Each vertex initially knows its unique identifier $\id(v)\in\N$, its local input, and its incident communication ports, but not $|V|$ or the maximum degree.
Ports at each vertex are numbered locally.
In each synchronous round, vertices send finite messages of arbitrary size through their ports, receive the incoming messages, and perform unrestricted finite local computation.
A randomized algorithm also gives each vertex private fair coins, independently of all other vertices.
The algorithms considered here prescribe an effective finite coin budget from each vertex's initial local input.
An $R$-round algorithm terminates at every vertex within $R$ rounds, for every input graph, every identifier assignment and every outcome of these coins.

The \emph{$r$-view} of $v$ consists of the rooted induced graph $G[B_G(v,r)]$, the port labels at both ends of its edges, and the initial local state of every vertex in the ball.
In particular, it records the full degree and port interface of boundary vertices, including ports leading outside the ball.
For a randomized execution it also records the prescribed coin strings of the vertices in the ball.
The output of an $r$-round algorithm at $v$ is determined by this view, by induction on the rounds~\cite{Linial92,Peleg00}.
Conversely, exchanging identifiers and then gathering the labelled local states lets $v$ learn its $r$-view in $r+1$ rounds.
Thus any effectively computable function of a fixed-radius $r$-view can be evaluated in $r+1$ rounds.
If $\dist_G(u,v)\le d$, the $(r+d)$-view of $v$ determines the $r$-view of~$u$.

A dominating-set algorithm outputs one bit at each vertex, and the vertices whose output is $1$ must form a dominating set.
A fractional dominating-set algorithm outputs nonnegative rational numbers $x_v$ that satisfy all domination constraints.
For a specified benchmark $\OPT\in\{\gamma,\gf\}$, a deterministic algorithm has ratio $\rho$ on $\Cc$ if its output has weight at most $\rho\cdot \OPT(G)$ for every $G\in\Cc$ and every identifier assignment.
A randomized algorithm has expected ratio $\rho$ if every outcome is feasible and its expected weight is at most $\rho\cdot \OPT(G)$.
\pagebreak

We use the following known result for the chosen minor-closed class.
To obtain the result we can use the algorithm of Czygrinow, Ha{\'n}{\'c}kowiak, Wawrzyniak and Witkowski~\cite[Theorem~1]{CHWW18} for graphs excluding a fixed topological minor: exclusion of a $K_t$ minor implies exclusion of a topological $K_t$.

\begin{fact}\label{fact:const-minor}
For every proper minor-closed class $\Cc$ there are constants $T_\Cc\ge0$ and $c_\Cc\ge1$ and a deterministic LOCAL algorithm that, on every $G\in\Cc$, computes a dominating set of size at most $c_\Cc\gamma(G)$ in $T_\Cc$ rounds.
The algorithm does not require the number of vertices or the maximum degree as part of the input.
\end{fact}

For planar graphs, we use the following explicit choice of constants.

\begin{theorem}[\cite{LOW08,LPW13,Waw14}]\label{thm:const-approx}
There are absolute constants $T_0$ and $c_0=52$ and a deterministic $T_0$-round LOCAL algorithm that computes, on every planar graph $G$, a dominating set $D_0$ with
$|D_0|\le c_0\gamma(G)$.
The algorithm requires neither the number of vertices nor the maximum degree as part of the input.
\end{theorem}

Any deterministic constant-round constant-factor approximation on the chosen class can serve as the initial subroutine.
For planar graphs, this includes the later algorithms in~\cite{HSV22,HKOSV25,Waw26}.
The choice affects our round bound but not the approximation ratio.


\section{The rounding lemma}\label{sec:rounding}

In this section we recall the standard threshold-rounding argument
of~\cite{BU17,Dvo19,HKOSV25}, which turns a fractional dominating
set $x$ into a dominating set of size at most $(2a+1)|x|$
on an $a$-sparse graph.
We include the short proof for completeness, using fractional
orientations to allow any real sparsity bound $a>0$.
We then record the consequence for distributed algorithms:
rounding adds only one communication round, provided that
each vertex can decide locally whether its coordinate meets
the threshold.

\begin{lemma}[threshold rounding~\cite{BU17,Dvo19,HKOSV25}]\label{lem:rounding}
Let $a>0$, let $G=(V,E)$ be $a$-sparse, and let $x$ be a fractional dominating set of~$G$.
Set $\theta=1/(2a+1)$ and
\[
D_1=\{v\in V:x_v\ge\theta\},
\qquad
D_2=V\setminus N[D_1].
\]
Then $D_1\cup D_2$ dominates $G$ and
\[
|D_1\cup D_2|\le(2a+1)|x|.
\]
If each vertex can decide whether $x_v\ge\theta$, the set is computed in one communication round.
In particular, this holds when $a$ is a known positive rational and the entries of $x$ are rational.
\end{lemma}

\begin{proof}
The vertices outside $N[D_1]$ are precisely those added to $D_2$, so $D_1\cup D_2$ dominates~$G$.
Each vertex first decides whether it belongs to $D_1$ and sends that bit to its neighbors.
After one round, it also knows whether it belongs to $D_2$.

By the definition of $D_1$,
$|D_1|\le\theta^{-1}x(D_1)$.
We prove that $|D_2|\le\theta^{-1}x(V\setminus D_1)$.
There is nothing to show if $D_2=\emptyset$, so assume otherwise and fix a fractional orientation with out-degree at most~$a$, as provided by \cref{lem:orientation}.
For $u\in D_2$, every vertex of $N[u]$ has weight less than~$\theta$.
Write
\[
\sigma(u)=\sum_{e=uv\ni u}x_v o_e(v).
\]
Splitting each $x_v$ according to the two orientations of $uv$, feasibility at $u$ gives
\begin{align*}
1
\le x_u+\sum_{e=uv\ni u}x_v o_e(u)+\sigma(u)
<\theta+\theta\sum_{e\ni u}o_e(u)+\sigma(u)
\le(a+1)\theta+\sigma(u).
\end{align*}
Since $1-(a+1)\theta=a\theta$, we have $\sigma(u)>a\theta$.
Every neighbor of a vertex in $D_2$ lies outside $D_1$.
Moreover, when we sum over $u\in D_2$, the contribution charged to a fixed vertex $v$ uses distinct incident edges.
Consequently,
\begin{align*}
a\theta|D_2|
<\sum_{u\in D_2}\sigma(u)
=\sum_{v\in V\setminus D_1}x_v
   \sum_{\substack{e=vu\in E\\u\in D_2}}o_e(v)
\le\sum_{v\in V\setminus D_1}x_v\sum_{e\ni v}o_e(v)
\le a\,x(V\setminus D_1).
\end{align*}
Dividing by $a\theta$ yields the required bound for $D_2$.
As $D_1$ and $D_2$ are disjoint, adding the two bounds finishes the proof.
\end{proof}

\begin{corollary}\label{cor:gap}
For every $a>0$ and every $a$-sparse graph $G$,
\[
\gamma(G)\le(2a+1)\gf(G).
\]
In particular, $\gamma(G)\le7\gf(G)$ for every planar graph~$G$.
\end{corollary}

\begin{corollary}\label{cor:det-rounding}
Let $a>0$, let $\Cc$ be a class of $a$-sparse graphs, and let $\OPT\in\{\gamma,\gf\}$.
Suppose a deterministic $R$-round algorithm computes, on every $G\in\Cc$, a fractional dominating set of weight at most $\rho\cdot\OPT(G)$.
If the threshold comparisons in \cref{lem:rounding} are locally computable, then a deterministic $(R+1)$-round algorithm computes a dominating set of size at most $(2a+1)\rho\cdot\OPT(G)$.
In particular, this holds when $a$ is a known positive rational.
\end{corollary}


\section{Randomized clustering}\label{sec:cluster}

We now describe the randomized clustering step, adapting the constructions of~\cite{CH06,CHW08,LR15}.
Starting from singleton parts, we repeatedly merge parts by selecting stars at random.
Every quotient remains in the minor-closed class, so the same density bound applies at every level.
We show that the total weight of edges between parts decreases geometrically in expectation, while the radius of each part remains bounded in terms of the number of levels.
In the next section, we apply this partition to the graph obtained from an initial dominating set.
Solving the resulting cluster subproblems gives the randomized approximation scheme.
Boundary repairs are used only in the comparison with a global optimum.

Throughout this section and the next, we fix a proper minor-closed class $\Cc$ containing a graph with at least one edge, and put $a=a(\Cc)$.
Then $a\ge1/2$, and every graph in $\Cc$ is $a$-sparse.
Set $q=1-1/(8a)$, so $0<q<1$.
The input $H\in\Cc$ has nonnegative rational edge weights $w$ and distinct vertex identifiers.
Fix an integer $L\ge1$.
For every $y\in V(H)$ and $t\in\{1,\ldots,L\}$, let $\kappa_t(y)\in\{\HEAD,\TAIL\}$ be a fair coin, where
all these coins are mutually independent.
After level~$t$, the procedure maintains a partition $\Pp_t$ of $V(H)$ and a leader $\lead(P)\in P$ for each part~$P$.
The identifier of a part is its leader's identifier, and its coin at level $t$ is $\kappa_t(\lead(P))$.

\begin{algorithm}[t]
\caption{$\textsc{RandomCluster}(H,w,L)$}\label{alg:cluster}
\begin{algorithmic}[1]
\State $\Pp_0\gets\{\{y\}:y\in V(H)\}$ and $\lead(\{y\})\gets y$
\For{$t=1,\dots,L$}
  \State $Q\gets H/\Pp_{t-1}$, with edge weights summed under contraction
  \State $\id(P)\gets\id(\lead(P))$ for every $P\in\Pp_{t-1}$
  \For{every non-isolated vertex $P$ of $Q$}
    \State choose a heaviest incident edge $e_P=Ph(P)$, breaking ties by the identifier of $h(P)$
  \EndFor
  \State $A\gets\{(P,h(P)): \deg_Q(P)>0\}$
  \State for each reciprocal pair $(P,P'),(P',P)$, delete the arc whose tail has the larger identifier
  \State $\beta(P)\gets\kappa_t(\lead(P))$ for every $P\in\Pp_{t-1}$
  \State $S\gets\{(P,P')\in A:\beta(P)=\TAIL,\ \beta(P')=\HEAD\}$
  \State $\Pp_t\gets\emptyset$
  \For{every $P'$ with $\beta(P')=\HEAD$}
    \State add $P'\cup\bigcup_{(P,P')\in S}P$ to $\Pp_t$, with leader $\lead(P')$
  \EndFor
  \State add each remaining $\TAIL$ part to $\Pp_t$, retaining its leader
\EndFor
\State \Return $\Pp_L$ and its leaders
\end{algorithmic}
\end{algorithm}

At each level of \cref{alg:cluster}, every non-isolated part chooses a heaviest incident edge in the quotient graph and points to its other endpoint.
If two parts choose each other, we keep only the arc directed from the smaller identifier to the larger one.
We then read the leaders' coins for this level and select the kept arcs directed from a \TAIL\ part to a \HEAD\ part.
Each \HEAD\ part absorbs the tails of its selected incoming arcs and all other parts remain unchanged.
\medskip

The heaviest-edge selection is the pseudo-forest construction of Czygrinow, Ha{\'n}{\'c}kowiak and Wawrzyniak~\cite[Fact~2.2]{CHW08}.
Their procedure \textsc{HeavyStars} applies the Cole-Vishkin algorithm~\cite{CV86} to obtain a proper $3$-coloring in $\Oof(\log^*n)$ rounds, and then uses the coloring to choose stars of large total weight.
Here one independent coin per part replaces that coloring step.
The resulting bound holds in expectation, which suffices for the approximation analysis.

To bound the radius of the parts, set $r_t=(3^t-1)/2$. Hence we have $r_0=0$ and $r_t=3r_{t-1}+1$. The following lemma shows the standard structural properties of the clustering process.
We include the short inductive proof to make the radius bound
and the dependence on the random choices explicit.

\begin{lemma}\label{lem:structure}
For every $t\in\{0,\ldots,L\}$ and every coin outcome, the following hold.
\begin{enumerate}[label=(\alph*)]
\item\label{it:partition} The family $\Pp_t$ partitions $V(H)$. Each part contains its leader, and distinct parts have distinct leaders.
\item\label{it:radius} Every vertex of a part $P\in\Pp_t$ has distance at most $r_t$ from $\lead(P)$ in $H[P]$. Consequently, every part is connected and $H/\Pp_t$ is a simple graph in $\Cc$.
\item\label{it:measurable} The partition $\Pp_t$ and its leaders depend only on $H$, its weights and identifiers, and the coins $\kappa_s(y)$ with $s\le t$.
\end{enumerate}
\end{lemma}

\begin{proof}
We argue by induction on $t$, with the singleton partition giving the base case.
At level $t$, each part has at most one outgoing arc in $A$.
A selected arc has a \TAIL\ tail and a \HEAD\ head, so no vertex can be both a tail and a head of selected arcs.
Thus $S$ is a collection of vertex-disjoint stars, directed towards their \HEAD\ centres.
Merging each star and retaining the parts incident with no selected arc gives a partition.
Every merged part inherits the leader of its centre and every unchanged part keeps its leader.
These leaders lie in their respective parts and are distinct by induction, proving~\ref{it:partition}.

For~\ref{it:radius}, consider a merged part with centre $P'$ and a leaf $P_i$.
Choose an edge $uv$ of $H$ with $u\in P_i$ and $v\in P'$.
For $y\in P_i$, the inductive hypothesis gives a walk from $y$ to $u$ through $\lead(P_i)$ of length at most $2r_{t-1}$, entirely within $P_i$.
Appending $uv$ and a path from $v$ to $\lead(P')$ gives a walk of length at most $3r_{t-1}+1=r_t$ inside the merged part.
Deleting repetitions gives the required path.
Vertices of the centre, and vertices of unchanged parts, retain the bound $r_{t-1}\le r_t$.
All parts are therefore connected, so their simple quotient is a minor of $H$ and belongs to $\Cc$.
Finally, level $t$ uses only the preceding partition, the fixed input and the level-$t$ coins, which proves~\ref{it:measurable}.
\end{proof}

The following lemma extends the heaviest-edge selection
argument of~\cite[Fact~2.2]{CHW08} from planar graphs to
$a$-sparse graphs.
We use fractional orientations to allow any real $a>0$.

\begin{lemma}\label{lem:heavy}
Let $a>0$, and let $Q$ be a simple $a$-sparse graph with nonnegative edge weights and distinct vertex identifiers.
The arc set $A$ constructed at one level of \cref{alg:cluster} has pairwise distinct underlying edges and satisfies
\[
 w(A)\ge\frac{w(Q)}{2a}.
\]
\end{lemma}

\begin{proof}
Let $I=\{P\in V(Q):\deg_Q(P)>0\}$ and $F=\{e_P:P\in I\}$.
Whenever both arcs $(u,v)$ and $(v,u)$ are present, we delete
one of them. The remaining arc set contains exactly one arc
for every edge of $F$.
Hence the arcs in $A$ have distinct underlying edges and $w(A)=w(F)$.
Since an edge can be chosen by at most its two endpoints,
\[
 \sum_{P\in I}w(e_P)\le2w(F)=2w(A).
\]
By \cref{lem:orientation}, choose a fractional orientation $o$ of $Q$ with out-degree at most $a$.
For every edge $e=PP'$, both $e_P$ and $e_{P'}$ have weight at least $w(e)$, so
\[
 w(e)\le o_e(P)w(e_P)+o_e(P')w(e_{P'}).
\]
Summing over the edges and using the out-degree bound gives
\[
 w(Q)\le\sum_{P\in I}w(e_P)\sum_{e\ni P}o_e(P)
 \le a\sum_{P\in I}w(e_P)
 \le2a\,w(A).\qedhere
\]
\end{proof}

For planar graphs, we may take $a=3$, giving $w(A)\ge w(Q)/6$.
This recovers the estimate of Czygrinow, Ha{\'n}{\'c}kowiak and Wawrzyniak~\cite[Fact~2.2]{CHW08}, whose proof uses the decomposition of a planar graph into three forests~\cite{NashWilliams64}.
\medskip

The next lemma gives the standard one-step estimate for randomized star contraction, as in the analysis of~\cite[Section~3]{LR15}.
We include the proof for our arc-selection rule and state the bound conditionally on the preceding levels. 
Recall that $\cut_w(\Pp)$ denotes the total weight of edges joining distinct parts of~$\Pp$.

\begin{lemma}\label{lem:one-level}
Let $\mathcal F_{t-1}$ be the information given by all coins $\kappa_s(y)$ with $s<t$.
For every $t\in\{1,\ldots,L\}$,
\[
 \E\bigl[\cut_w(\Pp_t)\mid\mathcal F_{t-1}\bigr]
 \le q\cut_w(\Pp_{t-1}).
\]
\end{lemma}

\begin{proof}
Fix the outcomes of all coins used in levels $1,\ldots,t-1$.
By \cref{lem:structure}\ref{it:measurable}, the preceding partition, its quotient $Q$ and the arc set $A$ are now fixed, with $w(Q)=\cut_w(\Pp_{t-1})$.
Distinct parts have distinct leaders, and the level-$t$ coins are independent of all earlier coins.
Thus the variables $\beta(P)$ are still independent fair coins under this conditioning.
Each arc of $A$ is selected with probability $1/4$.
Every selected arc becomes internal to a new part, and distinct selected arcs represent disjoint sets of original edges by \cref{lem:heavy}.
Consequently,
\[
 \cut_w(\Pp_t)\le w(Q)-w(S).
\]
The quotient $Q$ belongs to $\Cc$ and is therefore $a$-sparse by \cref{lem:structure}\ref{it:radius}.
Linearity of expectation and \cref{lem:heavy} therefore give
\[
 \E[w(S)\mid\mathcal F_{t-1}]
 =\tfrac14w(A)\ge\frac{w(Q)}{8a},
\]
which proves the claim.
\end{proof}

\begin{corollary}\label{cor:cut}
The final partition satisfies
\[
 \E[\cut_w(\Pp_L)]\le q^Lw(H).
\]
\end{corollary}
\begin{proof}
Initially, every edge joins different parts, so $\cut_w(\Pp_0)=w(H)$.
Taking expectations in \cref{lem:one-level} gives
$\E[\cut_w(\Pp_t)]\le q\E[\cut_w(\Pp_{t-1})]$.
Iterating this inequality $L$ times proves the bound.
\end{proof}

We next verify that each vertex can reconstruct its entire
part from a bounded neighborhood.
The argument follows the recursive reconstruction used
in~\cite[Section~4]{LR15}.
We give an explicit radius bound for our implementation,
including the information needed to determine the leaders
and the relevant coin outcomes.

\begin{lemma}\label{lem:cluster-local}
Set $\rho_t=2(3^t-1)$.
For every $t\in\{0,\ldots,L\}$ and $y\in V(H)$, the $\rho_t$-view of $y$, including the edge weights and coins, determines the entire part $P_t(y)\in\Pp_t$ containing $y$ and its leader.
In particular, \cref{alg:cluster} can be implemented in $\rho_L+1$ rounds on $H$.
\end{lemma}

\begin{proof}
We use induction on $t$. For $t=0$, the part is $\{y\}$.
Let $t\ge1$, write $r=r_{t-1}$, and put $P=P_{t-1}(y)$.
Every previous part has diameter at most $2r$ in $H$.
For each $z\in B_H(y,6r+3)$, its previous part and leader are determined by its $\rho_{t-1}$-view, which is contained in the $(\rho_{t-1}+6r+3)$-view of $y$.
That leader lies within distance $7r+3$ of $y$.
Hence the $\rho_t$-view reveals all these parts and their level-$t$ coins, since $\rho_t=\rho_{t-1}+8r+4$.

It remains to see that this information determines the new part of $y$.
The old part $P$ lies in $B_H(y,2r)$, each neighbor $C$ of $P$ in the quotient lies in $B_H(y,4r+1)$, and each quotient neighbor $R$ of such a $C$ lies in $B_H(y,6r+2)$.
For any such $R$, all edges leaving $R$ have their other endpoint in $B_H(y,6r+3)$.
Knowing the previous parts of these endpoints therefore determines every incident quotient edge and its weight, and hence the choice $h(R)$ whenever $R$ is non-isolated.
The same applies to $P$ and $C$.
For an incoming candidate $R\to C$, the choices~$h(R)$ and $h(C)$ determine whether the arc survives deletion of reciprocal pairs. The two leader coins then determine whether it is selected.
If $P$ is \HEAD, this identifies all parts absorbed by $P$.
If $P$ is \TAIL, it first determines whether $P$ joins its chosen neighbor $C$ and, if so, identifies every other part absorbed by $C$.
An isolated or unmatched part remains unchanged.
We have therefore recovered $P_t(y)$ and its leader.
By the view-gathering observation in \cref{sec:prelim},
each vertex can learn its $\rho_L$-view in $\rho_L+1$ rounds
and then compute its part and leader locally.
\end{proof}

\section{The randomized algorithm}\label{sec:rand}

We now use the clustering procedure to obtain randomized
algorithms for integral and fractional domination.
Both variants solve the problem independently inside each
cluster and return a feasible solution for every coin outcome.
We bound their expected weight by comparing the cluster
solutions with a global optimum and accounting for the
additional cost at the cluster boundaries.

Let $G=(V,E)\in\Cc$, and write $c=c_\Cc$ for the approximation factor of the initial algorithm in \cref{fact:const-minor}.
\Cref{alg:rand} first constructs a common random partition and then solves either the integral or the fractional dominating-set problem inside each cluster.
The class $\Cc$ and the constants of the initial algorithm are fixed throughout.

\begin{algorithm}[t]
\caption{$\textsc{RandDom}(L)$ on $G\in\Cc$}\label{alg:rand}
\begin{algorithmic}[1]
\State Compute $D_0$ using \cref{fact:const-minor}.
\State Set $\ell(v)=v$ for $v\in D_0$; for $u\notin D_0$, let $\ell(u)$ be the smallest-identifier vertex of $N(u)\cap D_0$.
\State Set $W_v=\ell^{-1}(v)$ for each $v\in D_0$.
\State Form the simple quotient $H=G/\{W_v:v\in D_0\}$ and assign weight $1$ to every edge of $H$.
\State Identify each quotient vertex $W_v$ with $v$, retaining its identifier and coins.
\State Run $\textsc{RandomCluster}(H,\mathbf1,L)$ using the coins $\kappa_1(v),\ldots,\kappa_L(v)$.
\State For each final part $P$, set $U_P=\bigcup_{v\in P}W_v$.
\For{each cluster $U$}
  \State \emph{Integral variant:} choose a minimum dominating set $D_U$ of $G[U]$.
  \State \emph{Fractional variant:} choose an optimal solution $y^U$ of $\LP(U)$.
  \State Break ties lexicographically in increasing identifier order.
\EndFor
\State \Return $y=\mathbf1_{\bigcup_U D_U}$ in the integral variant, or $y|_U=y^U$ in the fractional variant.
\end{algorithmic}
\end{algorithm}

For a cluster $U$, let $\LP(U)$ denote the standard
dominating-set linear program on $G[U]$.
In the fractional variant, we choose its lexicographically
smallest optimal solution in increasing identifier order.
This solution is rational and can be computed by
first finding the minimum total weight and then minimizing
the coordinates in increasing identifier order, while
preserving that total weight and all previously fixed
coordinates.
In the integral variant, we find a minimum dominating set
of $G[U]$ by exhaustive search.
Once each vertex knows its cluster and the induced graph,
these computations require no further communication.
By construction, every vertex $v\in D_0$ belongs to and
leads its own group $W_v$, and dominates every vertex
of that group.
In the proof of \cref{lem:per-outcome}, we use these leaders
to repair the restrictions of a global optimum to the
individual clusters.

\pagebreak
We record the basic properties of the groups and of the
clusters obtained by lifting the partition of $H$ back to $G$.

\begin{lemma}\label{lem:groups}
  \mbox{}
\begin{enumerate}[label=(\alph*)]
\item The groups $W_v$, for $v\in D_0$, partition $V$. Each contains its leader $v$, and $v\in N[u]$ for every $u\in W_v$. Thus $G[W_v]$ contains a spanning star centered at $v$.
\item The quotient $H$ is a simple graph in $\Cc$, with $k=|D_0|$ vertices and at most $ak$ edges.
\item For every final part $P$, the cluster $U_P$ is connected, and every vertex of $U_P$ is at distance at most $3r_L+1$ from $\lead(P)$ in $G[U_P]$.
\end{enumerate}
\end{lemma}

\begin{proof}
For (a), the domination property of $D_0$ makes $\ell$ well-defined, and its definition gives $\ell(u)\in N[u]\cap D_0$ and $\ell(v)=v$ for $v\in D_0$.
For (b), each group is connected by (a), so $H$ is a simple minor of $G$ and therefore belongs to $\Cc$. Its edge bound follows from the definition of $a=a(\Cc)$.
For (c), an edge $vv'$ of $H$ comes from an edge $uu'$ of $G$ with $u\in W_v$ and $u'\in W_{v'}$.
The group stars give a walk of length at most three from $v$ to $v'$, contained in $W_v\cup W_{v'}$.
Now take $x\in U_P$ and its group leader $v$.
First move from $x$ to $v$ in at most one step, and then lift a $v$-$\lead(P)$ path of length at most $r_L$ in $H[P]$, given by \cref{lem:structure}\ref{it:radius}.
The resulting walk lies in $G[U_P]$ and has length at most $1+3r_L$. Deleting repetitions gives the required path.
\end{proof}

\medskip
Fix a coin outcome.
We call a group $W_v$ a \emph{boundary group} if an edge of $G$ joins it to a different cluster.
Let $B\subseteq D_0$ be the set of leaders of these groups, and let $\cut(\Pp_L)$ count the edges of $H$ joining distinct final parts.
Because the initial edge weights of $H$ are one, $\cut(\Pp_L)=\cut_w(\Pp_L)$.

\begin{lemma}\label{lem:boundary}
For every coin outcome we have $|B|\le2\cdot \cut(\Pp_L)$.
\end{lemma}
\begin{proof}
A leader $v$ belongs to $B$ exactly when it is an endpoint of an edge of $H$ joining distinct final parts.
Each such edge has two endpoints, so $|B|\le2\cdot \cut(\Pp_L)$.
\end{proof}

The integral case of the following lemma is the
boundary-repair argument used in the proof
of~\cite[Theorem~3.4]{CHW08}.
We include the proof for both the integral and the
fractional variant.

\begin{lemma}\label{lem:per-outcome}
For every coin outcome, the following hold.
\begin{enumerate}[label=(\alph*)]
\item The integral variant returns a dominating set of $G$, and the fractional variant returns a fractional dominating set of $G$.
\item The integral output has weight at most $\gamma(G)+|B|$, and the fractional output has weight at most $\gf(G)+|B|$.
\end{enumerate}
\end{lemma}

\begin{proof}
For (a), let $U$ be the unique cluster containing $u$.
In the integral variant, $D_U$ dominates~$u$ already in $G[U]$.
In the fractional variant, $y(N[u])\ge y(N[u]\cap U)\ge1$.
Thus the union of the cluster solutions is feasible in either case.

For (b), first let $z$ be an optimal fractional dominating set of $G$ and set $z'=z+\mathbf1_B$.
We claim that $z'|_U$ is feasible for every cluster LP.
Take $u\in U$, and let $W_v\subseteq U$ be its group.
If $N[u]\subseteq U$, feasibility follows from $z(N[u])\ge1$.
Otherwise $W_v$ is a boundary group, so $v\in B$.
By \cref{lem:groups}(a), $v$ belongs to $N[u]\cap U$, and hence $z'(N[u]\cap U)\ge z'_v\ge1$.
Optimality within each cluster now gives
\[
 |y|=\sum_U|y^U|\le\sum_Uz'(U)
 =|z'|=\gf(G)+|B|.
\]
For the integral variant, take a minimum dominating set $D^*$ of $G$.
The same two cases show that $(D^*\cup B)\cap U$ dominates $G[U]$ for every cluster $U$.
Therefore
\[
 |y|=\sum_U|D_U|
 \le\sum_U|(D^*\cup B)\cap U|
 =|D^*\cup B|\le\gamma(G)+|B|.\qedhere
\]
\end{proof}

We now put these estimates together to bound the expected
weight and round complexity of \cref{alg:rand}.

\begin{theorem}\label{thm:rand}
For every integer $L\ge1$, \cref{alg:rand} runs in at most $T_\Cc+9\cdot3^L$ rounds and uses at most $L$ independent fair coins at each vertex.
The integral variant always returns a dominating set and satisfies
\[
 \E|y|\le\bigl(1+2acq^L\bigr)\gamma(G).
\]
The fractional variant always returns a fractional dominating set and satisfies
\[
 \E|y|\le\bigl(1+2a(2a+1)cq^L\bigr)\gf(G).
\]
\end{theorem}

\begin{proof}
The feasibility follows from \cref{lem:per-outcome}(a).
Writing $k=|D_0|$, we obtain from \cref{lem:boundary,cor:cut,lem:groups} that
\[
 \E|B|\le2\E[\cut(\Pp_L)]
 \le2q^L|E(H)|
 \le2akq^L
 \le2acq^L\gamma(G).
\]
Combining this with \cref{lem:per-outcome}(b) proves the integral estimate.
For the fractional estimate, use in addition $\gamma(G)\le(2a+1)\gf(G)$ from \cref{cor:gap}.

We next bound the number of rounds.
Membership in $D_0$ is determined by the $T_\Cc$-view, so $\ell(u)$ is determined by the $(T_\Cc+1)$-view of $u$, and a leader's group $W_v$ by its $(T_\Cc+2)$-view.
The edge-lifting argument in the proof of \cref{lem:groups} shows that $\dist_G(v,v')\le3\dist_H(v,v')$ for group leaders.
To reconstruct the $\rho$-view of a leader $v$ in $H$, it suffices to identify the groups of leaders at $G$-distance at most $3\rho+3$, and the edges between these groups.
This information, together with their coins, is contained in the $(3\rho+T_\Cc+5)$-view of $v$ in $G$.
By \cref{lem:cluster-local}, a vertex $u\in W_v$ can therefore identify its entire cluster $U$ from its $(3\rho_L+T_\Cc+6)$-view.
Every vertex of $U$ is at distance at most $2(3r_L+1)$ from $u$, by \cref{lem:groups}(c).
A view of radius
$2(3r_L+1)+3\rho_L+T_\Cc+6
 =T_\Cc+9\cdot3^L-1$ 
thus contains both the cluster information and the identifier-labelled graph $G[U]$ needed to apply the fixed optimization rule.
Gathering this view takes at most $T_\Cc+9\cdot3^L$ rounds.
\end{proof}

Choosing $L$ sufficiently large as a function of $\Cc$ and
$\eps$ proves \cref{thm:main}(a).
In the next section, we obtain the deterministic guarantees
by taking coordinatewise expectations of the fractional
output and applying the rounding lemma.
We collect all three guarantees in \cref{thm:minor-frac}.


\section{Derandomization}\label{sec:derand}

We now remove the randomness using Kuhn's derandomization
principle for distributed linear programs~\cite[Theorem~2.32]{Kuh05}.
Each vertex replaces its output coordinate with its expected value.
Since the domination constraints are linear, the resulting
vector is a fractional dominating set, and its weight equals
the expected weight of the randomized output.
Under our finite-randomness assumptions, each vertex can
compute this expectation exactly from a local view.
We first state the principle for fractional domination and
then apply it to \cref{alg:rand}.
Together with threshold rounding from \cref{lem:rounding},
this completes the proof of \cref{thm:main}.

\begin{theorem}\label{thm:derand}
Let $\Cc$ be a graph class and let $\OPT\in\{\gamma,\gf\}$.
Suppose that $\mathcal A$ is a randomized $R$-round LOCAL algorithm whose output at every vertex is rational. At each vertex, the algorithm uses at most an effectively prescribed finite number of independent fair coins, determined by its initial local input.
Assume that, for every $G\in\Cc$, every identifier assignment and every outcome of the coins, the output $y$ is a fractional dominating set, and that $\E|y|\le\rho\cdot\OPT(G)$.
Then the following hold.
\begin{enumerate}[label=(\alph*)]
\item In $R+1$ deterministic rounds, every vertex $v$ can compute $x_v=\E[y_v]$. The resulting vector~$x$ is a fractional dominating set with $|x|\le\rho\cdot\OPT(G)$.
\item If every graph in $\Cc$ is $a$-sparse and comparisons of rational numbers with $1/(2a+1)$ can be performed exactly, then a deterministic $(R+2)$-round algorithm outputs a dominating set of size at most $(2a+1)\rho\cdot\OPT(G)$.
\end{enumerate}
\end{theorem}

\begin{proof}
We recall Kuhn's argument~\cite[Theorem~2.32]{Kuh05},
making the exact computation of each expectation explicit.

For (a), fix a vertex $v$ and let $\mathcal V_v$ be its
$R$-view with the random strings omitted.
This view includes the initial local inputs and therefore
determines the prescribed coin budget at every vertex in
the view.
Let $m$ be the sum of these budgets, padding unused portions
of the random strings if necessary.
For each assignment $c\in\{0,1\}^m$, simulate the local rule
of $\mathcal A$ and denote its output at $v$ by
$f(\mathcal V_v,c)$.
Since coins outside this view cannot affect $y_v$, we have
\begin{align*}
x_v:=2^{-m}\sum_{c\in\{0,1\}^m}f(\mathcal V_v,c)
=\E[y_v].
\end{align*}
This is a finite average of rational numbers, so $v$ can
compute it exactly after learning $\mathcal V_v$ in $R+1$
rounds.

Every output of $\mathcal A$ is a fractional dominating set.
Hence $x\ge0$, and linearity of expectation gives
$x(N[u])=\E[y(N[u])]\ge1$ for every vertex $u$.
It also gives $|x|=\E|y|\le\rho\cdot\OPT(G)$.
Thus~$x$ has the required feasibility and weight.

For (b), apply \cref{lem:rounding}.
This takes one additional round and returns a dominating
set of size at most $(2a+1)|x|$, as claimed.
\end{proof}


Combining \cref{thm:rand,thm:derand} yields the following theorem
and completes the proof of our main result, \cref{thm:main}; for $\eps>1$ it
suffices to run the algorithm with accuracy~$1$.

\begin{theorem}\label{thm:minor-frac}\label{thm:frac}
Let $\Cc$ be a proper minor-closed class, put $a=a(\Cc)$, and let $\eps\in(0,1]$.
If $a=0$, both domination problems can be solved exactly in zero rounds by assigning value~$1$ to every vertex.
Otherwise put
\[
c=c_\Cc,\qquad q=1-\frac1{8a},\qquad \kappa_a=\frac{\ln3}{\ln(1/q)}.
\]
Then the following algorithms exist.
\begin{enumerate}[label=(\alph*)]
\item A randomized LOCAL algorithm runs in $T_\Cc+O_{a,c,\eps}(1)$ rounds, always outputs a dominating set, and has expected output size at most $(1+\eps)\gamma(G)$ on every $G\in\Cc$.
\item A deterministic LOCAL algorithm runs in $T_\Cc+O\bigl((2a(2a+1)c/\eps)^{\kappa_a}\bigr)$ rounds and outputs a rational fractional dominating set of weight at most $(1+\eps)\gf(G)$ on every $G\in\Cc$.
\item A deterministic LOCAL algorithm runs in $T_\Cc+O_{a,c,\eps}(1)$ rounds and outputs a dominating set of size at most $(2a+1+\eps)\gf(G)$ on every $G\in\Cc$.
\end{enumerate}
None of these algorithms requires the number of vertices or the maximum degree as part of the input.
\end{theorem}

\begin{proof}
The case $a=0$ is immediate, so assume $a>0$.
Minor closure then gives $K_2\in\Cc$, whence $a\ge1/2$ and $q\in(0,1)$.
Choose
\[
L=\left\lceil\frac{\ln(2a(2a+1)c/\eps)}{\ln(1/q)}\right\rceil.
\]
Since $a\ge1/2$, $c\ge1$, and $\eps\le1$, we have $L\ge1$.
Moreover,
\[
2a(2a+1)cq^L\le\eps,
\qquad
3^L\le3\bigl(2a(2a+1)c/\eps\bigr)^{\kappa_a}.
\]
By \cref{thm:rand}, the integral variant of \cref{alg:rand} always returns a dominating set and has expected output size at most $(1+2acq^L)\gamma(G)\le(1+\eps)\gamma(G)$.
This proves~(a).

The fractional variant always returns a fractional dominating set $y$ with $\E|y|\le(1+\eps)\gf(G)$ in at most $T_\Cc+9\cdot3^L$ rounds.
Its cluster solutions are rational, and its randomness uses at most $L$ fair coins per vertex.
Thus \cref{thm:derand}(a) computes the vector $x_v=\E[y_v]$ in at most
\[
T_\Cc+9\cdot3^L+1
\le T_\Cc+27\bigl(2a(2a+1)c/\eps\bigr)^{\kappa_a}+1
\]
rounds.
This proves~(b).

For~(c), choose a rational $a'$ with $a<a'<a+\eps/4$ and apply~(b) with accuracy $\eta=\eps/(2(2a'+1))$.
Every graph in $\Cc$ is $a'$-sparse, and the vector from~(b) is rational.
Rounding at the rational threshold $1/(2a'+1)$ by \cref{lem:rounding} takes one additional round and gives
\[
|D|\le(2a'+1)(1+\eta)\gf(G)
=\bigl(2a'+1+\eps/2\bigr)\gf(G)
\le(2a+1+\eps)\gf(G).
\]
The choice of $a'$ avoids requiring exact comparisons with the possibly unspecified real number~$a$.
All parameters are fixed for the class and the requested accuracy, so none of the algorithms needs the graph order or maximum degree.
\end{proof}

For planar graphs this gives $a=3$, $q=23/24$, and the factor $7+\eps$ in \cref{cor:planar-intro}.
\medskip

\begin{remark}\label{rem:parameters}
The theorem is stated for each fixed class and each fixed accuracy.
The required integer number of clustering levels and rational rounding threshold can therefore be built into the algorithm.
If, for every positive rational $\delta$, one can compute a rational $a_\delta$ with $a(\Cc)\le a_\delta<a(\Cc)+\delta$, these parameters can be chosen uniformly from a positive rational input~$\eps$.
In particular, the algorithms are uniform in~$\eps$ when the density is a specified rational number, as it is for planar graphs with $a=3$.
The existence of the algorithms for an arbitrary fixed class does not assume that its density can be computed from a description of the class.
\end{remark}

\section{Rounding with an additive excess}\label{sec:excess}

The bound in \cref{thm:minor-frac}(c) depends on the largest edge density in the class.
Suppose instead that $|E(G)|\le d|V(G)|+b$ throughout the class, where $b$ is a fixed constant.
We show that threshold rounding then costs a factor of $2d+1$ and an additive error depending on~$b$.
The additive-error reduction of Bonamy et al.~\cite[full version, Proposition~D.2]{BDGPSW26} then gives a multiplicative guarantee against $\gamma(G)$, with the asymptotic density in place of the supremum density.
For bounded genus, this lets us use the linear coefficient~$3$ in Euler's inequality.

For $d>0$, define $\Phi_d(G)=\max_{S\subseteq V(G)}\bigl(|E(G[S])|-d|S|\bigr)$.
The empty set shows that $\Phi_d(G)\ge0$, and $\Phi_d(G)=0$ exactly when $G$ is $d$-sparse.
More generally, if every subgraph $G'$ of~$G$ satisfies $|E(G')|\le d|V(G')|+b$, then $\Phi_d(G)\le b$.

The following lemma extends \cref{lem:orientation} and follows from a standard max-flow min-cut argument.
It also follows from Hajek's characterization of balanced allocations~\cite[Theorem~2]{Haj90}.
We include the proof for completeness.

\begin{lemma}\label{lem:excess-orientation}
For every graph $G$ and every $d>0$, there is a fractional orientation $o$ whose total out-degree excess above $d$ satisfies
\begin{align*}
\sum_{u\in V(G)}\max\{0,\mathrm{out}_o(u)-d\}\le\Phi_d(G),
\end{align*}
where $\mathrm{out}_o(u)=\sum_{e\ni u}o_e(u)$.
\end{lemma}

\begin{proof}
Write $G=(V,E)$.
Construct a flow network with a source $s$, a sink $t$, one node for each edge of $G$, and one node for each vertex of $G$.
For each edge $e=uv$, add an arc $s\to e$ of capacity~$1$ and arcs $e\to u$ and $e\to v$ of capacity $|E|+1$.
For each vertex $u$, add an arc $u\to t$ of capacity~$d$.

We show that every $s$--$t$ cut has capacity at least $|E|-\Phi_d(G)$.
A cut of capacity greater than~$|E|$ already satisfies this bound, since $\Phi_d(G)\ge0$.
For any other cut, let $S\subseteq V$ and $Y\subseteq E$ be the sets of vertex nodes and edge nodes on the source side.
No arc from an edge node to an endpoint crosses this cut, so $Y\subseteq E(G[S])$.
The cut therefore has capacity
\begin{align*}
|E\setminus Y|+d|S|
&\ge |E|-\bigl(|E(G[S])|-d|S|\bigr)
\ge |E|-\Phi_d(G).
\end{align*}
By the max-flow min-cut theorem, there is a flow $f$ of value $|f|\ge |E|-\Phi_d(G)$.

For each edge $e=uv$, assign the amounts $f(e\to u)$ and $f(e\to v)$ to its endpoints.
Their sum is $f(s\to e)\le1$.
Assign the remaining amount $1-f(s\to e)$ to either endpoint.
This defines a fractional orientation~$o$.
The amounts assigned by the flow contribute at most $d$ at each vertex, so only the remaining amounts can contribute to the excess.
Consequently,
\begin{align*}
\sum_{u\in V}\max\{0,\mathrm{out}_o(u)-d\}
&\le \sum_{e\in E}\bigl(1-f(s\to e)\bigr)
=|E|-|f|
\le\Phi_d(G),
\end{align*}
as required.
\end{proof}

We now extend \cref{lem:rounding} by allowing nonzero out-degree excess.
The proof follows the same rounding argument, using \cref{lem:excess-orientation} to bound the additional cost.
When $G$ is $d$-sparse, we have $\Phi_d(G)=0$ and recover the earlier bound.

\begin{lemma}[rounding with excess]\label{lem:rounding-excess}
Let $d>0$, let $G=(V,E)$ be a graph, and let $x$ be a fractional dominating set of~$G$.
Set $\theta=1/(2d+1)$, $D_1=\{v\in V:x_v\ge\theta\}$, and $D_2=V\setminus N[D_1]$.
Then $D_1\cup D_2$ is a dominating set and
\begin{align*}
|D_1\cup D_2|\le(2d+1)|x|+\frac2d\Phi_d(G).
\end{align*}
If each vertex can compare its value $x_v$ exactly with $\theta$, the set is computable in one LOCAL round.
\end{lemma}

\begin{proof}
The set $D_1\cup D_2$ dominates $G$, exactly as in \cref{lem:rounding}.
Under the stated comparison assumption, membership is decided in one round.
Also, $|D_1|\le(2d+1)x(D_1)$.
If $D_2=\emptyset$, this proves the claim, so assume $D_2\ne\emptyset$.

Choose the orientation from \cref{lem:excess-orientation}.
Write $\delta_u=\max\{0,\mathrm{out}_o(u)-d\}$ and
$\sigma(u)=\sum_{e=uv\ni u}x_v o_e(v)$.
Then $\mathrm{out}_o(u)\le d+\delta_u$ and $\sum_u\delta_u\le\Phi_d(G)$.
For $u\in D_2$, every vertex in $N[u]$ has value less than $\theta$.
Feasibility therefore gives
\begin{align*}
1
\le x_u+\sum_{e=uv\ni u}x_v o_e(u)+\sigma(u)
<\theta+\theta\mathrm{out}_o(u)+\sigma(u)
&\le(d+1)\theta+\theta\delta_u+\sigma(u).
\end{align*}
Thus $\sigma(u)>d\theta-\theta\delta_u$.
Summing over $D_2$ and collecting terms at their other endpoints, all of which lie outside $D_1$, yields
\begin{align*}
d\theta|D_2|-\theta\sum_{u\in D_2}\delta_u
<\sum_{v\in V\setminus D_1}x_v\mathrm{out}_o(v)
\le d\,x(V\setminus D_1)
   +\theta\sum_{v\in V\setminus D_1}\delta_v.
\end{align*}
Each sum of excesses is at most $\Phi_d(G)$.
Dividing by $d\theta$ gives
$|D_2|<(2d+1)x(V\setminus D_1)+2\Phi_d(G)/d$.
Adding the bound for $D_1$ proves the result.
\end{proof}

For this threshold rule, some allowance for excess is necessary on general graphs.
Indeed, on~$K_8$ the vector $x\equiv1/8$ is feasible and has weight~$1$.
With $d=3$, the threshold is $1/7$, so $D_1=\emptyset$ and $D_2=V(K_8)$.
The output has $8$ vertices, exceeding $7|x|$, while $\Phi_3(K_8)=28-24=4$.

We now combine \cref{lem:rounding-excess} with our fractional approximation and the additive-error reduction of~\cite[full version, Proposition~D.2]{BDGPSW26}.
This gives the bound in terms of asymptotic density stated in \cref{thm:asym-intro}.

\begin{corollary}\label{cor:asym}
For every proper minor-closed class $\Cc$ and every $\eps\in(0,1]$, there is a deterministic $O_{\Cc,\eps}(1)$-round LOCAL algorithm that, on every $G\in\Cc$, computes a dominating set of size at most $(2d^*(\Cc)+1+\eps)\gamma(G)$.
\end{corollary}

\begin{proof}
Choose a rational $d$ with $d^*(\Cc)<d<d^*(\Cc)+\eps/4$.
By the definition of $d^*(\Cc)$, there is a constant $b\ge0$ such that $|E(G)|\le d|V(G)|+b$ for every $G\in\Cc$.
Since $\Cc$ is closed under taking subgraphs, $\Phi_d(G)\le b$ for every $G\in\Cc$.

Use the fractional algorithm of \cref{thm:minor-frac}(b) with accuracy $\eps'=\eps/(4(2d+1))$.
Its output~$x$ is rational and has weight at most $(1+\eps')\gf(G)$.
The threshold $1/(2d+1)$ is rational, so \cref{lem:rounding-excess} applies in one additional round and returns a dominating set $D$ with
\begin{align*}
|D|
\le(2d+1)(1+\eps')\gf(G)+\frac{2b}{d}
\le(2d+1)(1+\eps')\gamma(G)+\frac{2b}{d}.
\end{align*}

The algorithm uses no graph-wide input parameters.
Its execution on a connected component~$K$ is therefore identical to its execution on $K$ alone with the same identifiers.
As $K\in\Cc$, the same guarantee holds separately on every component.

Apply the additive-error reduction of~\cite[full version, Proposition~D.2]{BDGPSW26} componentwise, with
$\alpha=(2d+1)(1+\eps')$, $\beta=2b/d$, and error parameter $\eps/4$.
If $b=0$, no reduction is needed.
In either case, the resulting dominating set $D'$ satisfies
\begin{align*}
|D'|
&\le\bigl((2d+1)(1+\eps')+\eps/4\bigr)\gamma(G)
=(2d+1+\eps/2)\gamma(G)
\le(2d^*(\Cc)+1+\eps)\gamma(G),
\end{align*}
where we sum the componentwise guarantees and use the additivity of $\gamma$.
All parameters depend only on $\Cc$ and $\eps$, so the number of rounds is $O_{\Cc,\eps}(1)$.
\end{proof}

\section{Sharp bounds and further consequences}\label{sec:beyond}

We now apply the general theorems to several familiar graph classes.
The uniform density bound in \cref{thm:minor-frac}(c) gives guarantees against the fractional optimum, including the planar bound~$7+\eps$.
For bounded Euler genus, the asymptotic-density bound removes the dependence of the approximation ratio on the genus, with the guarantee measured against the integral optimum.

\subsection{Planar graphs}\label{sec:planar}

Planar graphs have uniform density $3$, so \cref{thm:minor-frac}(c) gives the factor $7+\eps$ directly.
Using the initial approximation from \cref{thm:const-approx} also gives an explicit bound on the number of rounds.

\begin{corollary}\label{cor:planar}
For every $\eps\in(0,1]$, there is a deterministic LOCAL algorithm that returns, on every planar graph $G$, a dominating set of size at most $(7+\eps)\gf(G)$ in
\[
T_0+O\bigl((c_0/\eps)^\kappa\bigr)
\quad\text{rounds},\qquad
\kappa=\frac{\ln3}{\ln(24/23)}<25.9,
\]
where $c_0=52$ and $T_0$ are as in \cref{thm:const-approx}.
The algorithm requires no global parameter other than $\eps$, and no planar embedding.
\end{corollary}

\subsection{Subclasses of planar graphs}\label{sec:subclasses}
For a planar subclass, we may run the fractional algorithm of \cref{thm:minor-frac}(b) on the whole class of planar graphs and use a better sparsity bound only in the final rounding step.
This distinction matters for triangle-free graphs and graphs of prescribed girth: these classes are closed under taking subgraphs, but need not be minor-closed.

\begin{corollary}\label{cor:subclasses}
Let $\eps\in(0,1]$, and let $\Cc$ be a class of $a$-sparse planar graphs. Assume that comparisons of rational numbers with $1/(2a+1)$ can be performed exactly.
A deterministic LOCAL algorithm computes, on every $G\in\Cc$, a dominating set of size at most $(2a+1+\eps)\gf(G)$ in $T_0+O_{\eps,a}(1)$ rounds.
In particular, the approximation ratios are
\begin{itemize}[nosep]
\item $5+\eps$ for triangle-free planar graphs, including bipartite planar graphs;
\item $5+\eps$ for outerplanar graphs;
\item $(3g-2)/(g-2)+\eps$ for planar graphs of girth at least $g\ge4$, giving $13/3+\eps$ for girth at least $5$ and $4+\eps$ for girth at least $6$.
\end{itemize}
\end{corollary}

\begin{proof}
Use \cref{thm:minor-frac}(b) on the class of planar graphs with accuracy $\eps/(2a+1)$ and apply \cref{lem:rounding} with the sparsity constant $a$ of the subclass. The output has size at most
\[
(2a+1)\left(1+\frac{\eps}{2a+1}\right)\gf(G)
=(2a+1+\eps)\gf(G).
\]
The listed bounds follow from \cref{fact:sparse}; for girth at least $g$, the sparsity constant is $a=g/(g-2)$.
\end{proof}

For triangle-free planar graphs, this improves Wawrzyniak's ratio $6$~\cite{Waw25}. For these graphs, the square-grid lower bound~\cite[Section~4]{BCGW21}, together with our $5+\eps$ guarantee, determines the infimum of the deterministic constant-round approximation ratios as~$5$.

For approximation measured against $\gamma(G)$, our $5+\eps$ ratio does not improve the exact factor $5$ of Bonamy et al.~\cite{BCGW21}.
Bonamy et al.~\cite{BCGW21} also prove the matching lower bound $5-\eps$ for outerplanar graphs.

The girth-$5$ and girth-$6$ guarantees also improve the ratio $6$, which applies to both classes because they are triangle-free. For earlier work on planar graphs without $4$-cycles, see Alipour and Jafari~\cite{AJ20}. We do not obtain matching lower bounds for girth at least $5$ or $6$.

\subsection{Bounded Euler genus}\label{sec:genus}

The \emph{Euler genus} of an orientable surface with $h$ handles is $2h$, for a non-orientable surface with $k$ cross-caps it is $k$.
A graph has Euler genus at most $g$ if it embeds in a surface of Euler genus at most~$g$.
For each fixed $g$, these graphs form a minor-closed class~\cite{MT01}. We use the standard Euler edge bound for graphs on
surfaces~\cite[p.~928]{MW18} and its immediate consequence
for $\Phi_3(G)$.

\begin{lemma}\label{lem:genus-density}
Every graph $G$ of Euler genus at most $g\ge0$ with at least three vertices satisfies
\[
|E(G)|\le 3|V(G)|-6+3g.
\]
For every graph of Euler genus at most $g$, including graphs with fewer than three vertices,
\[
\Phi_3(G)\le3\max\{g-2,0\}.
\]
\end{lemma}

\begin{theorem}\label{thm:genus}
For every $g\ge0$ and $\eps\in(0,1]$, there is a deterministic LOCAL algorithm that runs in $O_{g,\eps}(1)$ rounds and, on every graph $G$ of Euler genus at most $g$, computes a dominating set of size at most $(7+\eps)\gamma(G)$.
If $g\le2$, the stronger bound $(7+\eps)\gf(G)$ holds.
\end{theorem}

\begin{proof}
Let $\Cc_g$ be the proper minor-closed class of graphs of Euler genus at most~$g$.
By \cref{lem:genus-density} and planar triangulations of unbounded order, $d^*(\Cc_g)=3$.
Applying \cref{cor:asym} gives the bound $(7+\eps)\gamma(G)$.
If $g\le2$, then \cref{lem:genus-density} implies that every graph in $\Cc_g$ is $3$-sparse.
Since $\Cc_g$ contains all planar graphs, $a(\Cc_g)=3$, so \cref{thm:minor-frac}(c) gives the stronger bound $(7+\eps)\gf(G)$.
\end{proof}

\begin{corollary}\label{cor:genus-optimal}
For every fixed $g\ge0$, the infimum of the approximation ratios achievable by deterministic constant-round LOCAL algorithms for minimum dominating set on graphs of Euler genus at most $g$ is~$7$.
\end{corollary}


\paragraph{Comparison.}
Amiri, Schmid and Siebertz gave a constant-round constant-factor algorithm for bounded-genus graphs~\cite{ASS16}. Its journal version has ratio $24g+\Oof(1)$~\cite{ASS19}.
For orientable surfaces, Czygrinow et al.~\cite{CHWW19} gave a constant-round algorithm with bounded messages whose guarantee includes an additive genus term. The LOCAL modification described by Bonamy et al.~\cite{BGPW25} gives ratio $91+\eps$.
Bonamy et al.~\cite{BGPW25,BDGPSW26} extended genus-independent guarantees to arbitrary fixed Euler genus.
Their transfer theorem converts a planar $\alpha$-approximation into a $(3\alpha+1)$-approximation~\cite[Theorem~4.2 and Corollary~4.3]{BDGPSW26}.
Together with the planar $(11+\eps)$-approximation of~\cite{HKOSV25}, this gave $34+\eps$ and their small-genus bound was $4\sqrt{3g/2}+14+\eps$.
Wawrzyniak's planar $(8+\eps)$-approximation subsequently gave $25+\eps$~\cite[Corollary~7.1]{Waw26}.
Applying the transfer theorem to \cref{cor:planar-intro} would give $22+\eps$. The required uniformity follows from~\cite[full version, Proposition~C.2]{BDGPSW26}.
The direct argument above gives $7+\eps$ for every fixed Euler genus, and \cref{cor:genus-optimal} shows that its limiting constant is best possible.
By comparison, \cref{thm:minor-frac}(c) alone gives $2a(\Cc_g)+1+\eps$, whose leading term grows with~$g$.

\subsection{Excluded minors and lower bounds from path powers}\label{sec:minors}

We next derive approximation guarantees for excluded-minor
classes and graphs of bounded treewidth and pathwidth.
To show that several of these guarantees are sharp, we first
prove a lower bound on powers of paths.
For each $r\ge1$, these graphs have pathwidth at most~$r$
and exclude $K_{r+2}$ as a minor, yet no deterministic
constant-round algorithm achieves an approximation ratio
below $2r+1$ on this family.
We then combine this lower bound with the relevant
edge-density estimates.
\medskip

For integers $r,n\ge1$, the \emph{$r$-th power} $P_n^r$ of the path has vertex set $\{1,\ldots,n\}$, with distinct vertices $i,j$ adjacent exactly when $|i-j|\le r$.

\begin{lemma}\label{lem:path-power}
Let $r\ge1$ and $n\ge r+1$ be integers.
\begin{enumerate}[label=(\alph*)]
\item $\gamma(P_n^r)\le\lceil n/(2r+1)\rceil$.
\item Both the treewidth and the pathwidth of $P_n^r$ are $r$. Consequently, $P_n^r$ has no $K_{r+2}$ minor.
\item For every $n\ge1$, the graph $P_n^2$ is outerplanar and $P_n^3$ is planar.
\end{enumerate}
\end{lemma}

\begin{proof}
(a) Put $q=\lceil n/(2r+1)\rceil$ and select the vertices $\min\{r+1+j(2r+1),n\}$ for $0\le j<q$.
Their closed neighborhoods cover the path in consecutive intervals, so they form a dominating set.

(b) The bags $\{i,\ldots,i+r\}$, for $1\le i\le n-r$, form a path decomposition of width~$r$.
The first $r+1$ vertices induce a clique, so the treewidth is at least~$r$.
Thus both widths equal~$r$.
Since treewidth is minor-monotone and $\operatorname{tw}(K_{r+2})=r+1$, there is no $K_{r+2}$ minor.

(c) The assertions are immediate for $n\le3$.
To obtain an outerplanar drawing of $P_k^2$ from one of $P_{k-1}^2$, place $k$ in the outer face next to the boundary edge $\{k-2,k-1\}$ and join it to both endpoints.
The edge $\{k-1,k\}$ is on the new outer boundary, so the construction continues inductively from the triangle $P_3^2$.
For $P_n^3$, start with a plane drawing of $K_4$ in which $2,3,4$ bound a face.
At step $k\ge5$, place $k$ inside the face bounded by $k-3,k-2,k-1$ and join it to these three vertices.
The triangle $k-2,k-1,k$ then bounds a face, allowing the next step.
\end{proof}

The following theorem rules out every deterministic constant-round
approximation ratio below $2r+1$ on $r$-th powers of paths,
even when all vertices know the order of the graph.

\begin{theorem}\label{thm:path-lb}
Let $r\ge1$ and $T\ge0$ be integers, let $\eps>0$, and let $\mathcal A$ be a deterministic $T$-round LOCAL algorithm.
There are an integer $n$ and an assignment of the identifiers $1,\ldots,n$ to $P_n^r$ such that the output of $\mathcal A$ either fails to dominate $P_n^r$ or has more than $(2r+1-\eps)\gamma(P_n^r)$ vertices.
The statement remains true when every vertex is given~$n$.
\end{theorem}

\begin{proof}
We use the Ramsey argument of~\cite{CHW08,HLS14} on a line.
Set $s=r(T+1)$ and $m=2s+1$.
Where the model has port numbers, fix them by ordering neighbors from left to right along the path.
For each $n\ge m$ and each $m$-element subset $Z\subseteq\{1,\ldots,n\}$, assign the elements of $Z$ in increasing order to positions $1,\ldots,m$ of $P_n^r$, and assign all remaining identifiers to the other positions arbitrarily.
Let $\chi_n(Z)$ be the output of $\mathcal A$ at position $s+1$, giving every vertex the correct input~$n$ if this input is allowed.
This defines a two-coloring of the $m$-element subsets.
The color is independent of the completion: every vertex within distance $T$ of the centre, together with all of its neighbors, lies among positions $1,\ldots,m$.
Moreover, every such vertex has its full degree $2r$ and the same port pattern as an interior vertex of a longer path power.
Thus the color describes the output for this ordered local pattern wherever it occurs away from the endpoints.

Fix an integer $M\ge2s+2r+1$.
Let $N\ge M$ be a Ramsey number such that every two-coloring of the $m$-element subsets of an $N$-element set has a monochromatic $M$-element subset.
For each $n\ge N$, apply this statement repeatedly to the coloring $\chi_n$, removing a monochromatic $M$-element set at each step until fewer than $N$ identifiers remain.
We obtain disjoint sets $Z_1,\ldots,Z_k$ with $k\ge(n-N)/M$.
Assign the identifiers of $Z_j$ in increasing order to the $j$-th consecutive block of $M$ vertices of $P_n^r$.
Place the remaining identifiers at the end, in any order.

Consider a vertex at position $i$ within block $j$, where $s<i\le M-s$.
Its padded local pattern consists of $m$ identifiers from $Z_j$ in increasing order, so its output is the common color $\chi_j$ of these subsets.
If $\chi_j=0$, the vertex at position $s+r+1$ in the block and all of its neighbors output zero.
The output would then fail to dominate the graph.
Consequently, every valid output contains at least $M-2s$ vertices from each full block.
Using \cref{lem:path-power}(a), its approximation ratio is at least
\[
\frac{k(M-2s)}{n/(2r+1)+1}
\ge
\frac{(n-N)(M-2s)/M}{n/(2r+1)+1}.
\]
For fixed $M$ and $N$, the right-hand side tends to $(2r+1)(1-2s/M)$ as $n\to\infty$.
Choose $M$ large enough that this limit exceeds $2r+1-\eps/2$, and then choose $n$ large enough that the displayed ratio exceeds $2r+1-\eps$.
\end{proof}

For $r=3$, the path powers are planar and have pathwidth three.
Thus \cref{thm:path-lb} recovers the planar lower bound $7-\eps$ of Hilke, Lenzen and Suomela~\cite{HLS14} on this more restricted family.
For $r=2$, it recovers the outerplanar lower bound $5-\eps$ of~\cite{BCGW21}.

\begin{corollary}\label{cor:minors}
For every $\eps>0$, there are deterministic constant-round LOCAL algorithms with the following approximation guarantees, all measured against $\gf(G)$:
\begin{enumerate}[label=(\alph*)]
\item $2t-3+\eps$ on $K_t$-minor-free graphs, for integers $3\le t\le9$;
\item $2a_t+1+\eps$ on $K_t$-minor-free graphs, for every integer $t\ge3$, where $a_t$ is the supremum of $|E(G)|/|V(G)|$ over nonempty $K_t$-minor-free graphs and
\[
a_t=(\alpha+o(1))t\sqrt{\ln t},\qquad \alpha=0.319\ldots
\]
as $t\to\infty$~\cite{Tho01};
\item $2r+1+\eps$ on graphs of treewidth at most $r$, for every integer $r\ge0$;
\item $t+2+\eps$ on $K_{2,t}$-minor-free graphs, for every integer $t\ge2$;
\item $7+\eps$ on $K_5$-minor-free graphs and on $K_{3,3}$-minor-free graphs.
\end{enumerate}
For ordinary minimum dominating set approximation, the infimum of the achievable deterministic constant-round ratios is $2t-3$ in (a), $2r+1$ for both treewidth and pathwidth at most $r$, and $7$ for each class in (e).
\end{corollary}

\begin{proof}
All the classes are proper minor-closed classes, so it suffices to apply \cref{thm:minor-frac}(c) with the appropriate density bound.

(a) For $3\le t\le7$, Mader's theorem~\cite{Mader68} gives $|E(G)|\le(t-2)|V(G)|$ for $K_t$-minor-free graphs; the cases $t\le5$ are due to Dirac, and $t=6$ was obtained independently by Gy\H{o}ri.
For $t=8$, J{\o}rgensen's theorem gives $|E(G)|\le6|V(G)|-20$ when $|V(G)|\ge8$~\cite{Jor94}; for $t=9$, Song and Thomas give $|E(G)|\le7|V(G)|-27$ when $|V(G)|\ge9$~\cite{ST06}.
These inequalities include the exceptional graphs in the two extremal theorems.
If $n<t$, the elementary bound $|E(G)|\le n(n-1)/2\le(t-2)n$ suffices.
Thus $a_t\le t-2$ throughout the stated range.

(b) Apply the definition of $a_t$ and Thomason's asymptotic formula~\cite{Tho01}.

(c) Every nonempty graph of treewidth at most $r$ has a vertex of degree at most $r$, and the same is true of each of its subgraphs.
Deleting such vertices successively gives $|E(G)|\le r|V(G)|$.

(d) Chudnovsky, Reed and Seymour~\cite{CRS11} prove that every nonempty $K_{2,t}$-minor-free graph satisfies
\[
|E(G)|\le\frac{t+1}{2}(|V(G)|-1).
\]
The empty graph is harmless, and the required sparsity bound is $(t+1)/2$.

(e) The $K_5$-minor-free case is contained in (a).
For $K_{3,3}$-minor-free graphs, we prove $|E(G)|\le3n-5$ for $n\ge2$ by induction on~$n$.
The assertion is immediate for $n\le3$.
An isolated vertex can be deleted.
If $G$ is disconnected or has a cut vertex and has no isolated vertices, write it as the union of two smaller graphs $G_1,G_2$ with no common edge, with $n_1,n_2\ge2$, and with $n_1+n_2\le n+1$.
Induction gives
\[
|E(G)|\le(3n_1-5)+(3n_2-5)\le3n-7\le3n-5.
\]
Now suppose that $G$ is $2$-connected and has a separating pair $\{u,v\}$.
Split the components of $G-\{u,v\}$ into two nonempty groups, and let $G_1,G_2$ be the corresponding sides, assigning the edge $uv$, if present, to $G_1$ only.
Then $n_1,n_2\ge3$ and $n_1+n_2=n+2$.
By $2$-connectivity, each side contains a $u$-$v$ path with an internal vertex and with all its internal vertices on that side.
Contracting such a path on the opposite side shows that $G_i+uv$ is a minor of $G$.
Hence
\[
|E(G)|\le|E(G_1+uv)|+|E(G_2+uv)|-1
\le3(n+2)-11=3n-5.
\]
Finally, if $G$ is $3$-connected, Hall's theorem~\cite{Hall43}, together with the forbidden-minor characterisation of planar graphs, implies that $G$ is planar or isomorphic to $K_5$.
In either case $|E(G)|\le3n-5$.

By \cref{lem:path-power}(b), $P_n^{t-2}$ is $K_t$-minor-free for $t\ge3$, and $P_n^r$ has both pathwidth and treewidth~$r$ for $r\ge1$, provided $n$ is sufficiently large.
The lower bounds in \cref{thm:path-lb} therefore give the claimed infima for (a) and (c).
For $r=0$, every vertex of an edgeless graph must be selected, so the optimum ratio is~$1$.
The upper bound for pathwidth follows from the treewidth bound.
Finally, both classes in (e) contain all planar graphs, so their infimum is at least $7$ by~\cite{HLS14}.
\end{proof}

Some of the smallest cases already have exact optimal ratios: forests admit ratio $3$, and $K_{2,t}$-minor-free graphs admit ratio $2t-1$ for $t=2,3$~\cite[Table~1 and Theorem~4.4]{BGPW25b}.
Thus the new corollary does not improve these cases. In particular, it gives only $4+\eps$ when $t=2$.
For every $t\ge3$, the class of $K_{2,t}$-minor-free graphs contains all outerplanar graphs, so its infimum is at least $5$~\cite{BCGW21}.
For $t=2$, the path powers $P_n^2$ with $n\ge4$ are unavailable because they contain a $4$-cycle.
Ordinary paths instead give the lower bound $3-\eps$ by \cref{thm:path-lb}, matching the known ratio~$3$.
\Cref{tab:beyond} shows the resulting bounds and the corresponding lower bounds.


\section{Conclusion and open problems}\label{sec:open}

Our upper bounds follow by combining and adapting established
techniques: cluster-and-solve algorithms~\cite{CH06,CHW08},
randomized star contraction~\cite{LR15}, Kuhn's derandomization
principle for distributed linear programs~\cite{Kuh05}, and
local threshold rounding~\cite{BU17,Dvo19,HKOSV25}, and the additive-error reduction of~\cite[full version, Proposition~D.2]{BDGPSW26}.
The main contribution is the approximation guarantees obtained
from this combination.
We analyse the fractional cluster solutions directly against
the global fractional optimum and extend the rounding analysis
to account for an additive density excess.

For every proper minor-closed class $\Cc$ and every $\eps>0$,
this yields a deterministic $(1+\eps)$-approximation for
fractional dominating set in a number of rounds depending
only on $\Cc$ and $\eps$.
Threshold rounding gives a dominating set of size at most
$(2a(\Cc)+1+\eps)\gf(G)$.
Accounting for density excess gives the further bound
$(2d^*(\Cc)+1+\eps)\gamma(G)$.

Together with the lower bounds from powers of paths, these
guarantees determine the infimum of the deterministic
constant-round approximation ratios for planar graphs,
graphs of Euler genus at most any fixed $g$, graphs of
treewidth or pathwidth at most any fixed $r\ge1$, and
$K_t$-minor-free graphs with $3\le t\le9$.
In particular, the planar infimum is~$7$.
Our $5+\eps$ upper bound and the square-grid lower bound~\cite[Section~4]{BCGW21} also determine the infimum for triangle-free planar graphs as $5$.
\medskip

The following questions remain open.

\begin{enumerate}[label=\arabic*.]
\item \emph{Excluded minors.} For $K_t$-minor-free graphs with $t\ge10$, the infimum lies between $2t-3$ and $2a_t+1$, where $a_t=\Theta(t\sqrt{\log t})$ is the supremum edge density. Is the lower bound $2t-3$ attainable as an infimum? Since the class is closed under disjoint unions, $d^*=a_t$, so the asymptotic-density refinement gives no further improvement. For $K_{2,t}$-minor-free graphs with $t\ge4$, the corresponding interval is $[5,\min\{t+2,50\}]$.
\item \emph{The exact ratio.} Does a deterministic constant-round LOCAL algorithm achieve ratio $7$ on every planar graph? Our round bound grows as $\eps\to0$, while the lower bound of~\cite{HLS14} excludes only ratios strictly below~$7$.
\item \emph{Round complexity and probability.} For a fixed proper minor-closed class, what is the optimal dependence on $\eps$ of the round complexity for an expected $(1+\eps)$-approximation, or for a deterministic fractional $(1+\eps)$-approximation? Under what lower bound on $\gamma(G)$ in terms of $n=|V(G)|$ can a constant-round algorithm achieve ratio $1+\eps$ with failure probability at most $n^{-c}$, for a prescribed constant $c>0$?
\item \emph{Efficient deterministic algorithms.} Can the exponent in our polynomial round bounds be substantially reduced, already for a deterministic $(7+\eps)$-approximation on planar graphs? Can the same guarantees be obtained with local computation polynomial in the size of each collected neighborhood?
\item \emph{Bounded messages.} Is a constant-round $(7+\eps)$-approximation possible in the CONGEST model, with $\Oof(\log n)$ bits per edge in each round?
\item \emph{Anonymous networks.} What is the optimal deterministic constant-round ratio on planar networks with port numbers but without unique identifiers? Our algorithms use identifiers to break ties. The identifier-based lower bound also applies to this weaker model.
\item \emph{Planar subclasses.} For each integer $g\ge5$, determine the infimum of the deterministic constant-round ratios on planar graphs of girth at least $g$; see \cref{cor:subclasses}.
\end{enumerate}

{\raggedright
\IfFileExists{plainurl.bst}{\bibliographystyle{plainurl}}{\bibliographystyle{plain}}
\bibliography{refs}
}

\end{document}